\documentclass[a4paper,onecolumn,11pt,accepted=2021-04-05]{quantumarticle}
\pdfoutput=1
\usepackage[utf8]{inputenc}
\usepackage[T1]{fontenc}

\usepackage{graphicx, booktabs, algorithm2e}
\usepackage{microtype}

\usepackage[numbers,sort&compress]{natbib}

\usepackage{amsmath, amssymb, amsthm}
\usepackage{bm, mathtools, textcomp}
\usepackage{caption}
\usepackage[dvipsnames]{xcolor}

\newtheorem*{definition}{Definition}
\newtheorem{theorem}{Theorem}
\newtheorem{proposition}{Proposition}
\newtheorem{lemma}{Lemma}
\newtheorem{claim}{Claim}

\newcommand{\C}{\mathbb{C}}
\newcommand{\R}{\mathbb{R}}

\DeclarePairedDelimiter\norm{\lVert}{\rVert}%

\def\ket#1{\mathinner{\left\lvert{#1}\right\rangle}}

\newcommand{\interior}{\operatorname{int}}

\newcommand{\lorentz}{\mathcal{L}}

\renewcommand{\vec}[1]{\bm{#1}}
\newcommand{\vecrest}[1]{\widetilde{\vec{#1}}}

\newcommand{\measured}[1]{\overline{#1}}
\newcommand{\scaled}[1]{#1'}
\newcommand{\measuredscaled}[1]{\scaled{\measured{#1}}}

\newcommand{\dx}{\Delta\vec{x}}
\newcommand{\dy}{\Delta\vec{y}}
\newcommand{\ds}{\Delta\vec{s}}
\newcommand{\dxs}{\scaled{\dx}}
\newcommand{\dss}{\scaled{\ds}}
\newcommand{\dxm}{\measured{\dx}}
\newcommand{\dym}{\measured{\dy}}
\newcommand{\dsm}{\measured{\ds}}
\newcommand{\dxms}{\measuredscaled{\dx}}
\newcommand{\dsms}{\measuredscaled{\ds}}

\newcommand{\xm}{\measured{\vec{x}}}
\newcommand{\sm}{\measured{\vec{s}}}
\renewcommand{\ss}{\scaled{\vec{s}}}
\newcommand{\xs}{\scaled{\vec{x}}}
\newcommand{\xnext}{\vec{x}_\text{next}}
\newcommand{\ynext}{\vec{y}_\text{next}}
\newcommand{\snext}{\vec{s}_\text{next}}
\newcommand{\xsnext}{\scaled{\xnext}}
\newcommand{\ssnext}{\scaled{\snext}}

\newcommand{\Arw}{\operatorname{Arw}}

\newcommand{\poly}{\operatorname{poly}}
\newcommand{\polylog}{\operatorname{polylog}}

\newcommand{\exponent}{2.591}
\newcommand{\exponentlb}{2.564}
\newcommand{\exponentub}{2.619}

\graphicspath{{Figures/}}
\usepackage{appendix}

\newcommand{\change}[1]{#1}

\title{Quantum algorithms for Second-Order Cone Programming and Support Vector Machines}

\author{Iordanis Kerenidis}
\affiliation{QCWare, Palo Alto, California}
\affiliation{Université de Paris, CNRS, IRIF, F-75006, Paris, France}
\orcid{0000-0003-0659-3727}
\email{jkeren@irif.fr}

\author{Anupam Prakash}
\affiliation{QCWare, Palo Alto, California}
\affiliation{Université de Paris, CNRS, IRIF, F-75006, Paris, France}
\orcid{0000-0002-6853-7125}
\email{anupamprakash1@gmail.com}

\author{D\'aniel Szil\'agyi}
\affiliation{Université de Paris, CNRS, IRIF, F-75006, Paris, France}
\orcid{0000-0002-8518-8422}
\email{dszilagyi@irif.fr}

\begin{document}
	\maketitle
	
	\begin{abstract}
	We present a quantum interior-point method (IPM) for second-order cone programming (SOCP) that runs in time $\widetilde{O} \left( n\sqrt{r} \frac{\zeta \kappa}{\delta^2} \log \left(1/\epsilon\right) \right)$ where $r$ is the rank and $n$ the dimension of the SOCP, $\delta$ bounds the distance of intermediate solutions from the cone boundary, $\zeta$ is a parameter upper bounded by $\sqrt{n}$, and $\kappa$ is an upper bound on the condition number of matrices arising in the classical IPM for SOCP. The algorithm takes as its input a suitable quantum description of an arbitrary SOCP and outputs a classical description of a $\delta$-approximate $\epsilon$-optimal solution of the given problem.
	
	Furthermore, we perform numerical simulations to determine the values of the aforementioned parameters when solving the SOCP up to a fixed precision $\epsilon$. We present experimental evidence that in this case our quantum algorithm exhibits a polynomial speedup over the best classical algorithms for solving general SOCPs that run in time $O(n^{\omega+0.5})$ (here, $\omega$ is the matrix multiplication exponent, with a value of roughly $2.37$ in theory, and up to $3$ in practice). For the case of random SVM 
(support vector machine) instances of size $O(n)$, the quantum algorithm scales as $O(n^k)$, where the exponent $k$ is estimated to be $2.59$ using a least-squares power law. 
On the same family random instances, the estimated scaling exponent for an external SOCP solver is $3.31$ while that for a state-of-the-art SVM solver is $3.11$. 
\end{abstract}
	\section{Introduction}
It is well known that many interesting and relevant optimization problems in the domain of Machine Learning can be expressed in the framework of convex optimization \cite{boyd2004convex,bubeck2015convex}. The landmark result in this area was the discovery of interior-point methods (IPM) by \cite{karmarkar1984new}, and their subsequent generalization to all ``self-scaled'' (i.e. symmetric) cones by \cite{nesterov1997self, nesterov1998primal}. Very recently, \cite{cohen2019solving} have shown that it is possible to solve linear programs (LP) in $\widetilde{O}(n^\omega)$, the time it takes to multiply two matrices (as long as $\omega \geq 2+1/6$, which is currently the case). This result has been further extended in \cite{lee2019solving} to a slightly more general class of cones, however, their techniques did not yield improved complexities for second-order (SOCP) and semidefinite programming (SDP). \change{An efficient algorithm for SOCP would also yield efficient algorithms for many interesting problems, such as (standard and quadratically-constrained) convex quadratic programming, portfolio optimization, and many others \cite{alizadeh2003second}.}

Starting with the landmark results of \cite{grover1996fast, shor1999polynomial}, and, more recently, \cite{harrow2009quantum}, it has been demonstrated that quantum computers offer significant (sometimes even exponential) asymptotic speedups for a number of important problems. More recently, there has been substantial work in the area of convex optimization. Quantum speedups for gradient descent were investigated by \cite{gilyen2019optimizing}, whereas \cite{brandao2017quantum, brandao2019quantum, van2017quantum, van2019improvements} presented quantum algorithms for SDP based on the the multiplicative weights framework of \cite{arora2012multiplicative}. However, it has been difficult to establish asymptotic speedups for this family of quantum SDP solvers as their the running time depends on problem-specific parameters, including a 5th-power dependence on the width of the SDP. \change{Interestingly, the recent result of \cite{brandao2019faster} suggests that such a speedup might be obtained when applying an SDP algorithm of this type to some low-precision instances of quadratic binary optimization.}

In an orthogonal approach, \cite{kerenidis2018quantum} proposed a quantum algorithm for LPs and SDPs by quantizing a variant of the classical interior point method and using the state of the art quantum linear algebra tools \cite{chakraborty2019power, gilyen2019quantum} -- in particular, the matrix multiplication and inversion algorithms whose running time is sub-linear in the input size. However, the complexity of this algorithm depends on the condition number of $O(n^2)$-sized matrices that is difficult to bound theoretically. It therefore remains an open question to find an end-to-end optimization problems for which quantum SDP solvers achieve an asymptotic speedup over state of the art classical algorithms. 
In this paper, we propose support vector machines (SVM) as a candidate for such an end-to-end quantum speedup using a quantum interior point method based algorithm.

\subsection{Our results and techniques}
In this section, we provide a high level sketch of our results and the techniques used for the quantum interior point method for SOCPs, we begin by discussing the differences between classical and quantum interior point methods. 

A classical interior point method solves an optimization problem over symmetric cones by starting with a feasible solution and iteratively finding solutions with a smaller duality gap while maintaining feasibility. A single iterative step consists of solving a system of linear equations called the Newton linear system and updating the current iterate using the solutions of the Newton system. The analysis of the classical IPM shows that in each iteration, the updated solutions remain feasible and the duality gap is decreased by a factor of $(1- \alpha/\sqrt{n})$ where $n$ is the dimension of the optimization problem and $\alpha > 0$ is a constant. The algorithm therefore converges to a feasible solution with duality gap $\epsilon$ in $O(\sqrt{n} \log (1/\epsilon))$ iterations. 

A quantum interior point method \cite{kerenidis2018quantum} uses a quantum linear system solver instead of classical one in each iteration of the IPM. However, there is an important difference between classical and quantum linear algebra procedures for solving the linear system $A\vec{x}=\vec{b}$. 
Unlike classical linear system solvers which return an exact description of $\vec{x}$, quantum tomography procedures can return an $\epsilon$-accurate solution $\widetilde{\vec{x}}$ such that $\norm{ \measured{\vec{x}} - \vec{x}} \leq \epsilon \norm{\vec{x}}$ with \change{$O(n/\epsilon^{2})$} runs of the quantum linear system solver. \change{Additionally, these linear system solvers require $A$ and $\vec{b}$ to be given as block-encodings~\cite{chakraborty2019power}, so this input model is used by our algorithm as well.} One of the main technical challenges in developing a quantum interior point method (QIPM) is to establish convergence of the classical IPM which uses $\epsilon$-approximate solutions of the Newton linear system (in the $\ell_{2}$ norm) instead of the exact solutions in the classical IPM. 

The quantum interior point method for second-order cone programs requires additional ideas going beyond those used for the quantum interior point methods for SDP \cite{kerenidis2018quantum}. Second order cone programs are optimization problems over the product of second-order or Lorentz cones (see section \ref{sec:socp} for definitions), interior point methods for SOCP can be described using the Euclidean Jordan algebra framework \cite{monteiro2000polynomial}. The Euclidean Jordan algebra framework provides analogs of concepts like eigenvalues, spectral and Frobenius norms for matrices and positive semidefinite constraints for the case of SOCPs.
Using these conceptual ideas from the Euclidean Jordan algebra framework \cite{monteiro2000polynomial} and the analysis of the approximate 
SDP interior point method \cite{kerenidis2018quantum} we provide an approximate IPM for SOCP that converges in $O(\sqrt{n} \log (1/\epsilon))$ iterations. Approximate IPMs for SOCP have not been previously investigated in the classical or the quantum optimization literature, this analysis is one of the main technical contributions of this paper.

From an algorithmic perspective, SOCPs are much closer to LPs (Linear Programs) than to SDPs, since for cones of dimension $n$, the Newton linear systems arising in LP and SOCP IPMs are of size $O(n)$, whereas in the SDP case they are of size $O(n^2)$. \change{Namely, a second-order conic constraint of dimension $n$ can be expressed as a single PSD constraint on a (sparse) $n \times n$ matrix \cite{alizadeh2003second} -- this allows us to embed an SOCP with $n$ variables and $m$ constraints in an SDP with an $n\times n$ matrix and $m$ constraints. The cost of solving that SDP would have a worse dependence on the error \cite{van2019improvements} or the input size \cite{kerenidis2018quantum}. On the other hand,} the quantum representations (block encodings) of the Newton linear systems for SOCP are also much simpler to construct than those for SDP. The smaller size of the SOCP linear system also makes it feasible to empirically estimate the condition number for these linear systems in a reasonable amount of time allowing us to carry out extensive numerical experiments to validate the running time of the quantum algorithm.

The theoretical analysis of the quantum algorithm for SOCP shows that its worst-case running time is
\begin{equation}\label{eq:complexity}
\widetilde{O} \left( n\sqrt{r} \frac{\zeta \kappa}{\delta^2} \log \left(\frac1\epsilon\right) \right),
\end{equation}	
where $r$ is the rank and $n$ the dimension of the SOCP, $\delta$ bounds the distance of intermediate solutions from the cone boundary, $\zeta$ is a parameter bounded by $\sqrt{n}$, $\kappa$ is an upper bound on the condition number of matrices arising in the interior-point method for SOCPs, and $\epsilon$ is the target duality gap. The running time of the algorithm depends on problem dependent parameters like $\kappa$ and $\delta$ that are difficult to bound in terms of the problem dimension $n$ -- this is also the case with previous quantum SDP solvers \cite{kerenidis2018quantum, van2019improvements} and makes it important to validate the quantum optimization speedups empirically. \change{Interestingly, since we require a classical solution of the Newton system, the linear system solver could also be replaced by a classical iterative solver \cite{saad2003iterative} which would yield a complexity of $O(n^2 \sqrt{r} \kappa \log(n/\epsilon))$.}

Let us make a remark about the complexity: as it is the case with all approximation algorithms, \eqref{eq:complexity} depends on the inverse of the target duality gap $\epsilon$, thus the running time of our algorithm grows to infinity as $\epsilon$ approaches zero, as in the case of classical IPM. Our running time also depends on $\kappa$, which in turn is empirically observed to grow inversely with the duality gap (in particular as $O(1/\epsilon)$) which again makes the running time go to infinity as $\epsilon$ approaches zero. The quantum IPM is a low precision method, unlike the classical IPM, and it can offer speedups for settings where the desired precision $\epsilon$ is moderate or low. \change{Thus although at first glance it seems that the $\epsilon$-dependence in \eqref{eq:complexity} is logarithmic, experimental evidence suggests that the factor $\frac{\zeta \kappa}{\delta^2}$ depends polynomially on $1/\epsilon$.}

Support Vector Machines (SVM) is an important application in machine learning, where even a modest value of $\epsilon=0.1$ yields an almost optimal classifier. Since the SVM training problem can be reduced to SOCP, the quantum IPM for SOCP can be used to obtain an efficient quantum SVM algorithm. We perform extensive numerical experiments to evaluate our algorithm on random SVM instances and compare it against state of the art classical SOCP and SVM solvers. 

The numerical experiments on random SVM instances indicate that the running time of the quantum algorithm scales as roughly $O(n^{\exponent})$, where all the parameters in the running time are taken into account and the exponent is estimated using a least squares fit. We also benchmarked the exponent for classical SVM algorithms on the same instances and for a comparable accuracy, the scaling exponent was found to be $3.31$ for general SOCP solvers and $3.11$ for state-of-the-art SVM solvers. We note that this does not prove a worst-case asymptotic speedup, but the experiments on unstructured SVM instances provide strong evidence for a significant polynomial speedup of the quantum SVM algorithm over state-of-the-art classical SVM algorithms. We can therefore view SVMs as a candidate problem for which quantum optimization algorithms can achieve a polynomial speedup over state of the art classical algorithms for an end-to-end application. 

\subsection{Related work}
Our main result is the first specialized quantum algorithm for training support-vector machines (SVM). While several quantum SVM algorithms have been proposed, they are unlikely to offer general speedups for the most widely used formulation of SVM -- the soft-margin ($\ell_1$-)SVM (defined in eq. \eqref{prob:SVM}). On one hand, papers such as \cite{arodz2019quantum} formulate the SVM as a SDP, and solve that using existing quantum SDP solvers such as \cite{brandao2017quantum,brandao2019quantum,van2017quantum,van2019improvements} -- with the conclusion being that a speedup is observed only for very specific sparse instances. On the other hand, \cite{rebentrost2014quantum} solves an easier related problem -- the least-squares SVM ($\ell_2$-SVM or LS-SVM, see eq. \eqref{prob:LS-SVM} for its formulation), thus losing the desirable sparsity properties of $\ell_1$-SVM \cite{suykens2002weighted}. It turns out that applying our algorithm to this problem also yields the same complexity as in \cite{rebentrost2014quantum} for the $\ell_2$-SVM. Very recently, a quantum algorithm for yet another variant of SVM (SVM-perf, \cite{joachims2006training}) has been presented in \cite{allcock2020quantum}.

\section{Preliminaries}
\subsection{Second-order cone programming} \label{sec:socp} 
\change{For the sake of completeness, in this section we present the most important results about classical SOCP IPMs, from \cite{monteiro2000polynomial,alizadeh2003second}. We start by defining SOCP as the} optimization problem over the product of second-order (or Lorentz) cones $\lorentz = \lorentz^{n_1} \times \cdots \times \lorentz^{n_r}$, where $\lorentz_k \subseteq \R^k$ is defined as $\lorentz^k = \left\lbrace\left. \vec{x} = (x_0; \vecrest{x}) \in \R^{k} \;\right\rvert\; \norm{\vecrest{x}} \leq x_0 \right\rbrace$. In this paper we consider the problem \eqref{prob:SOCP primal} and its dual \eqref{prob:SOCP dual}:
\begin{center}
\begin{tabular}{rl}
	\begin{minipage}{0.45\linewidth}
		\begin{equation}
			\begin{array}{ll}
				\min & \vec{c}^\top \vec{x}\\
				\text{s.t.}& A \vec{x} = \vec{b} \\
				& \vec{x} \in \lorentz,
			\end{array} \label{prob:SOCP primal}
		\end{equation}
	\end{minipage} 
	&
	\begin{minipage}{0.45\linewidth}
		\begin{equation}
			\begin{array}{ll}
				\max & \vec{b}^\top \vec{y}\\
				\text{s.t.}& A^\top \vec{y} + \vec{s} = \vec{c}\\
				& \vec{s} \in \lorentz,\; \vec{y} \in \R^m.
			\end{array} \label{prob:SOCP dual}
		\end{equation}
	\end{minipage}
\end{tabular}
\end{center}
We call $n:= \sum_{i=1}^r n_i$ the \emph{size} of the SOCP \eqref{prob:SOCP primal}, and $r$ is its \emph{rank}.

A solution $(\vec{x}, \vec{y}, \vec{s})$ satisfying the constraints of both \eqref{prob:SOCP primal} and \eqref{prob:SOCP dual} is \emph{feasible}, and if in addition it satisfies $\vec{x} \in \interior \lorentz$ and $\vec{s} \in \interior \lorentz$, it is \emph{strictly feasible}. If at least one constraint of \eqref{prob:SOCP primal} or \eqref{prob:SOCP dual}, is violated, the solution is \emph{infeasible}. The duality gap of a feasible solution $(\vec{x}, \vec{y}, \vec{s})$ is defined as $\mu := \mu(\vec{x}, \vec{s}) := \frac1r \vec{x}^\top \vec{s}$. As opposed to LP, and similarly to SDP, strict feasibility is required for \emph{strong duality} to hold \cite{alizadeh2003second}. From now on, we assume that our SOCP has a strictly feasible primal-dual solution (this assumption is valid, since the homogeneous self-dual embedding technique from  \cite{ye1994hsd} allows us to embed \eqref{prob:SOCP primal} and \eqref{prob:SOCP dual} in a slightly larger SOCP where this condition is satisfied). 

\subsection{Euclidean Jordan algebras}
The cone $\lorentz^n$ has an algebraic structure similar to that of symmetric matrices under the matrix product. Here, we consider the Jordan product of $(x_0, \vecrest{x}) \in \R^{n}$ and $(y_0, \vecrest{y}) \in \R^{n}$, defined as
\begin{equation*}
	\vec{x} \circ \vec{y} := \begin{bmatrix}
		\vec{x}^T \vec{y} \\
		x_0\vecrest{y} + y_0\vecrest{x}
	\end{bmatrix}, \text{ and its identity element }\vec{e} := \begin{bmatrix}
		1 \\
		0^{n-1}
	\end{bmatrix}.
\end{equation*}
This product is closely related to the (linear) matrix representation $\Arw(\vec{x}) := \begin{bmatrix}
	x_0 & \vecrest{x}^T \\
	\vecrest{x} & x_0 I_{n-1}
\end{bmatrix}$, which in turn satisfies the following equality:
\begin{equation*}
	\vec{x} \circ \vec{y} = \Arw(\vec{x}) \vec{y} = \Arw(\vec{x}) \Arw(\vec{y}) \vec{e}.
\end{equation*}
The key observation is that this product induces a spectral decomposition of any vector $\vec{x}$, that has similar properties as its matrix relative. Namely, for any vector $\vec{x}$ we define
\begin{align}\label{eq:jordan eigenvalues and eigenvectors}
\lambda_1(\vec{x}) &:= x_0 + \norm{\vecrest{x}},\; \vec{c}_1(\vec{x}) := \frac12 \begin{bmatrix}
1 \\
\frac{\vecrest{x}}{\norm{\vecrest{x}}}
\end{bmatrix}, \nonumber \\
\lambda_2(\vec{x}) &:= x_0 - \norm{\vecrest{x}},\; \vec{c}_2(\vec{x}) := \frac12 \begin{bmatrix}
1 \\
\frac{-\vecrest{x}}{\norm{\vecrest{x}}}
\end{bmatrix}.
\end{align}
\change{We use the shorthands $\lambda_1 := \lambda_1(\vec{x})$, $\lambda_2 := \lambda_2(\vec{x})$, $\vec{c}_1 := \vec{c}_1(\vec{x})$ and $\vec{c}_2 := \vec{c}_2(\vec{x})$ whenever $\vec{x}$ is clear from the context,} so we observe that $\vec{x} = \lambda_1 \vec{c}_1 + \lambda_2 \vec{c}_2$. The set of eigenvectors $\{\vec{c}_1, \vec{c}_2 \}$ is called the \emph{Jordan frame} of $\vec{x}$, and satisfies several properties:
\begin{proposition}[Properties of Jordan frames]
	Let $\vec{x} \in \R^{n}$ and let $\{ \vec{c}_1, \vec{c}_2 \}$ be its Jordan frame. Then, the following holds:
	\begin{enumerate}
		\item $\vec{c}_1 \circ \vec{c}_2 = 0$ (the eigenvectors are ``orthogonal'')
		\item $\vec{c}_1\circ \vec{c}_1 = \vec{c}_1$ and $\vec{c}_2 \circ \vec{c}_2 = \vec{c}_2$
		\item $\vec{c}_1$, $\vec{c}_2$ are of the form $\left( \frac12; \pm \vecrest{c} \right)$ with $\norm{\vecrest{c}} = \frac12$
	\end{enumerate}
\end{proposition}
On the other hand, just like a given matrix is positive (semi)definite if and only if all of its eigenvalues are positive (nonnegative), a similar result holds for $\lorentz^n$ and $\interior \lorentz^n$ (the Lorentz cone and its interior):
\begin{proposition}
	Let $\vec{x} \in \R^{n}$ have eigenvalues $\lambda_1, \lambda_2$. Then, the following holds:
	\begin{enumerate}
		\item $\vec{x} \in \lorentz^n$ if and only if $\lambda_1\geq 0$ and $\lambda_2 \geq 0$.
		\item $\vec{x} \in \interior \lorentz^n$ if and only if $\lambda_1 > 0$ and $\lambda_2 > 0$.
	\end{enumerate}
\end{proposition}
Now, using this decomposition, we can define arbitrary real powers $\vec{x}^p$ for $p \in \R$ as $\vec{x}^p := \lambda_1^p\vec{c}_1 + \lambda_2^p \vec{c}_2$, and in particular the ``inverse'' and the ``square root''
\begin{align*}
	\vec{x}^{-1} &= \frac{1}{\lambda_1} \vec{c}_1 + \frac{1}{\lambda_2} \vec{c}_2, \text{ if } \lambda_1\lambda_2 \neq 0, \\
	\vec{x}^{1/2} &= \sqrt{\lambda_1} \vec{c}_1 + \sqrt{\lambda_2} \vec{c}_2, \text{ if } \vec{x} \in \lorentz^n.
\end{align*}
Moreover, we can also define some operator norms, namely the Frobenius and the spectral one:
\begin{align*}
	\norm{\vec{x}}_F &= \sqrt{\lambda_1^2 + \lambda_2^2} = \sqrt{2}\norm{\vec{x}}, \\
	\norm{\vec{x}}_2 &= \max\{ |\lambda_1|, |\lambda_2| \} = |x_0| + \norm{\vecrest{x}}.
\end{align*}
Finally, we define an analogue to the operation $Y \mapsto XYX$. It turns out that for this we need another matrix representation (\emph{quadratic representation}) $Q_{\vec{x}}$, defined as
\begin{align} \label{qrep} 
	Q_{\vec{x}} &:= 2\Arw^2(\vec{x}) - \Arw(\vec{x}^2)
	= \begin{bmatrix}
	\norm{\vec{x}}^2 & 2x_0\vecrest{x}^T \\
	2x_0\vecrest{x} & \lambda_1\lambda_2 I_n + 2\vecrest{x}\vecrest{x}^T
	\end{bmatrix}.
\end{align}
Now, the matrix-vector product $Q_{\vec{x}}\vec{y}$ will behave as the quantity $XYX$. To simplify the notation, we also define the matrix $T_{\vec{x}} := Q_{\vec{x}^{1/2}}$.

The definitions that we introduced so far are suitable for dealing with a single constraint $\vec{x} \in \lorentz^n$. For dealing with multiple constraints $\vec{x}_1 \in \lorentz^{n_1}, \dots, \vec{x}_r \in \lorentz^{n_r}$, we need to deal with block-vectors $\vec{x} = (\vec{x}_1;\vec{x}_2;\dots;\vec{x}_r)$ and $\vec{y} = (\vec{y}_1;\vec{y}_2;\dots;\vec{y}_r)$. We call the number of blocks $r$ the \emph{rank} of the vector (thus, up to now, we were only considering rank-1 vectors). Now, we extend all our definitions to rank-$r$ vectors.
\begin{enumerate}
	\item $\vec{x} \circ \vec{y} := (\vec{x}_1 \circ \vec{y}_1; \dots; \vec{x}_r \circ \vec{y}_r)$
	\item The matrix representations $\Arw(\vec{x})$ and $Q_{\vec{x}}$ are the block-diagonal matrices containing the representations of the blocks:
	\begin{align*}
		\Arw(\vec{x}) &:= \Arw(\vec{x}_1) \oplus \cdots \oplus \Arw(\vec{x}_r) \text{ and }
		Q_{\vec{x}} := Q_{\vec{x}_1} \oplus \cdots \oplus Q_{\vec{x}_r}
	\end{align*}
	\item $\vec{x}$ has $2r$ eigenvalues (with multiplicities) -- the union of the eigenvalues of the blocks $\vec{x}_i$. The eigenvectors of $\vec{x}$ corresponding to block $i$ contain the eigenvectors of $\vec{x}_i$ as block $i$, and are zero everywhere else.
	\item The identity element is $\vec{e} = (\vec{e}_1; \dots ;\vec{e}_r)$, where $\vec{e}_i$'s are the identity elements for the corresponding blocks.
\end{enumerate}
Thus, all things defined using eigenvalues can also be defined for rank-$r$ vectors:
\begin{enumerate}
	\item The norms are extended as $\norm{\vec{x}}^2_F := \sum_{i=1}^r \norm{\vec{x}_i}_F^2$ and $\norm{\vec{x}}_2 := \max_i \norm{\vec{x}_i}_2$, and
	\item Powers are computed blockwise as $\vec{x}^p := (\vec{x}_1^p;\dots;\vec{x}_r^p)$ whenever the corresponding blocks are defined.
\end{enumerate}

Some further matrix-algebra inspired properties of block vectors are stated in the following two claims:
\begin{claim}[Algebraic properties]\label{claim:algebraic properties}
	Let $\vec{x}, \vec{y}$ be two arbitrary block-vectors. Then, the following holds:
	\begin{enumerate}
		\item The spectral norm is subadditive: $\norm{\vec{x} + \vec{y}}_2 \leq \norm{\vec{x}}_2 + \norm{\vec{y}}_2$.
		\item The spectral norm is less than the Frobenius norm: $\norm{\vec{x}}_2 \leq \norm{\vec{x}}_F$.
		\item If $A$ is a matrix with minimum and maximum singular values $\sigma_{\text{min}}$ and $\sigma_{\text{max}}$ respectively, then the norm $\norm{A\vec{x}}$ is bounded as $\sigma_{\text{min}} \norm{\vec{x}} \leq \norm{A\vec{x}} \leq \sigma_{\text{max}} \norm{\vec{x}}$.
		\item The minimum eigenvalue of $\vec{x} + \vec{y}$ is bounded as $\lambda_\text{min}(\vec{x} + \vec{y}) \geq \lambda_\text{min}(\vec{x}) - \norm{\vec{y}}_2$.
		\item The following submultiplicativity property holds: $\norm{\vec{x} \circ \vec{y}}_F \leq \norm{\vec{x}}_2 \cdot \norm{\vec{y}}_F$.
	\end{enumerate}
\end{claim}
\change{In general, the proofs of these statements are analogous to the matrix case, with a few notable differences:} First, the vector spectral norm $\norm{\cdot}_2$ is not actually a norm, since there exist nonzero vectors outside $\lorentz$ which have zero norm. It is, however, still bounded by the Frobenius norm (just like in the matrix case), which is in fact a proper norm. Secondly, the minimum eigenvalue bound also holds for matrix spectral norms, with the exact same statement. Finally, the last property is reminiscent of the matrix submultiplicativity property $\norm{A \cdot B}_F \leq \norm{A}_2 \norm{B}_F$.

We finish with several well-known properties of the quadratic representation $Q_{\vec{x}}$ and $T_{\vec{x}}$.
\begin{proposition}[Properties of $Q_{\vec{x}}$, from \cite{alizadeh2003second}]\label{claim:Properties of Qx} 
	Let $\vec{x} \in \interior \lorentz$. Then, the following holds:
	\begin{enumerate}
		\item $Q_{\vec{x}} \vec{e} = \vec{x}^2$, and thus $T_{\vec{x}} \vec{e} = \vec{x}$.
		\item $Q_{\vec{x}^{-1}} = Q_{\vec{x}}^{-1}$, and more generally $Q_{\vec{x}^p} = Q_{\vec{x}}^p$ for all $p \in \R$.
		\item $\norm{Q_{\vec{x}}}_2 = \norm{\vec{x}}_2^2$, and thus $\norm{T_{\vec{x}}}_2 = \norm{\vec{x}}_2$.
		\item $Q_{\vec{x}}$ preserves $\lorentz$, i.e. $Q_{\vec{x}}(\lorentz) = \lorentz$ and $Q_{\vec{x}}(\interior \lorentz) = \interior \lorentz$. 
	\end{enumerate}
\end{proposition}

\subsection{Interior-point methods}
Our algorithm follows the general IPM structure, i.e. it uses Newton's method to solve a sequence of increasingly strict relaxations of the Karush-Kuhn-Tucker (KKT) optimality conditions:
\begin{align}
	A\vec{x} = \vec{b},\; A^\top \vec{y} + \vec{s} = \vec{c} \label{eq:central path} \\
	\vec{x} \circ \vec{s} = \nu \vec{e},\;
	\vec{x} \in \lorentz, \vec{s} \in \lorentz, \nonumber
\end{align}
where the parameter $\nu>0$ is decreased by a factor $\sigma<1$ in each iteration. Since $\vec{x} \circ \vec{s} = \nu \vec{e}$ implies that the duality gap is $\mu = \nu$, by letting $\nu \to 0$ the IPM converges towards the optimal solution. The curve traced by (feasible) solutions $(\vec{x}, \vec{y}, \vec{s})$ of \eqref{eq:central path} for $\nu > 0$ is called the \emph{central path}, and we note that all points on it are strictly feasible.

More precisely, in each iteration we need to find $\dx, \dy, \ds$ such that $\xnext := \vec{x} + \dx$, $\ynext := \vec{y} + \dy$ and $\snext := \vec{s} + \ds$ satisfy \eqref{eq:central path} for $\nu = \sigma \mu$. After linearizing the product $\xnext\circ \snext$, we obtain the following linear system -- the \emph{Newton system}:
\begin{align}
\begingroup
\begin{bmatrix}
A & 0 & 0 \\
0 & A^\top & I \\
\Arw(\vec{s}) & 0 & \Arw(\vec{x}) 
\end{bmatrix}
\begin{bmatrix}
\dx \\
\dy \\
\ds
\end{bmatrix} = 
\begin{bmatrix}
\vec{b} - A \vec{x} \\
\vec{c} - \vec{s} - A^\top \vec{y} \\
\sigma \mu \vec{e} - \vec{x} \circ \vec{s}
\end{bmatrix}.
\endgroup
\label{eq:Newton system}
\end{align}
As a final remark, it is not guaranteed that $(\xnext, \ynext, \snext)$ is on the central path \eqref{eq:central path}, or even that it is still strictly feasible. Luckily, it can be shown that long as $(\vec{x}, \vec{y}, \vec{s})$ starts out in a neighborhood $\mathcal{N}$ of the central path, $(\xnext, \ynext, \snext)$ will remain both strictly feasible and in $\mathcal{N}$.
It can be shown that this algorithm halves the duality gap every $O(\sqrt{r})$ iterations, so indeed, after $O(\sqrt{r} \log(\mu_0 / \epsilon))$ it will converge to a (feasible) solution with duality gap at most $\epsilon$ (given that the initial duality gap was $\mu_0$).

\subsection{Quantum linear algebra}
As it was touched upon in the paper, the main speedup in our algorithms comes from the fact that we use the quantum linear algebra algorithms from \cite{chakraborty2019power, gilyen2019quantum} whose complexity is sublinear in the dimension. Of course, we need to change our computational model for this sentence to make any sense: namely, we encode $n$-dimensional unit vectors as quantum states of a $\lceil \log_2(n) \rceil$-qubit system. In other words, for a vector $\vec{z} \in \R^{2^k}$ with $\norm{\vec{z}} = 1$, we use the notation $\ket{\vec{z}}$ to refer to the $k$-qubit state $\ket{\vec{z}} := \sum_{i=0}^{2^k - 1} z_i \ket{i}$, where $\ket{i}$ are the standard basis vectors of $\C^{2^k}$, the state space of a $k$-qubit system \cite{nielsen2010quantum}. 

Given a quantum state $\ket{\vec{z}}$, we have very limited ways of interacting with it: we can either apply a unitary transformation $U: \C^{2^k} \to \C^{2^k}$, or we can \emph{measure} it, which means that we discard the state and obtain a single random integer $0 \leq i \leq 2^k - 1$, with the probability of measuring $i$ being $z_i^2$. In particular this means that we can neither observe the \emph{amplitudes} $z_i$ directly, nor can we create a copy of $\ket{\vec{z}}$ for an arbitrary $\ket{\vec{z}}$. In addition to this, it is \emph{a priori} not clear how (and whether it is even possible) to implement the state $\ket{\vec{z}}$ or an arbitrary unitary $U$ using the gates of a quantum computer. Luckily, there exists a quantum-classical framework using \emph{QRAM} data structures described in \cite{kerenidis2016quantum} that provides a positive answer to both of these questions.

The QRAM can be thought of as the quantum analogue to RAM, i.e. an array $[\vec{b}^{(1)}, \dots, \vec{b}^{(m)}]$ of $w$-bit bitstrings, whose elements we can access in poly-logarithmic time given their address (position in the array). More precisely, QRAM is just an efficient implementation of the unitary transformation 
\begin{equation*}
\ket{i}\ket{0}^{\otimes w} \mapsto \ket{i} \ket{b^{(i)}_1\dots b^{(i)}_w}, \text{ for } i \in [m].
\end{equation*}
The usefulness of QRAM data structures becomes clear when we consider the \emph{block encoding} framework:
\begin{definition}
	Let $A \in \R^{n\times n}$ be a symmetric matrix. Then, the $\ell$-qubit unitary matrix $U \in \C^{2^\ell \times 2^\ell}$ is a $(\zeta, \ell)$ block encoding of $A$ if $U = \begin{bmatrix}
	A / \zeta & \cdot \\
	\cdot & \cdot
	\end{bmatrix}$. For an arbitrary matrix $B \in \R^{n \times m}$, a block encoding of $B$ is any block encoding of its symmetrized version $\operatorname{sym}(B) := \begin{bmatrix}0 & B\\ B^\top & 0\end{bmatrix}$.
\end{definition}
We want $U$ to be implemented efficiently, i.e. using an $\ell$-qubit quantum circuit of depth (poly-)logarithmic in $n$. Such a circuit would allow us to efficiently create states $\ket{A_i}$ corresponding to rows (or columns) of $A$.  Moreover, we need to be able to construct such a data structure efficiently from the classical description of $A$. It turns out that we are able to fulfill both of these requirements using a data structure built on top of QRAM.
\begin{theorem}[Block encodings using QRAM \cite{kerenidis2016quantum, kerenidis2017quantum}] \label{qbe}
	There exist QRAM data structures for storing vectors $\vec{v}_i \in \R^n$, $i \in [m]$ and matrices $A \in \R^{n\times n}$ such that with access to these data structures one can do the following:
	\begin{enumerate}
		\item Given $i\in [m]$, prepare the state $\ket{\vec{v}_i}$ in time $\widetilde{O}(1)$. In other words, the unitary $\ket{i}\ket{0} \mapsto \ket{i}\ket{\vec{v}_i}$ can be implemented efficiently.
		\item A $(\zeta(A), 2 \log n)$ unitary block encoding for $A$ with $\zeta(A) = \norm{A}_{2}^{-1}\min( \norm{A}_{F}, s_{1}(A))$, where $s_1(A) = \max_i \sum_j |A_{i, j}|$ can be implemented in time $\widetilde{O}(\log n)$. Moreover, this block encoding can be constructed in a single pass over the matrix $A$, and it can be updated in $O(\log^2 n)$ time per entry.
	\end{enumerate}
\end{theorem}
From now on, we will also refer to storing vectors and matrices in QRAM, meaning that we use the data structure from Theorem \ref{qbe}. Note that this is the same quantum oracle model that has been used to solve SDPs in \cite{kerenidis2018quantum} and \cite{van2019improvements}. 

Once we have these block encodings, we may use them to perform linear algebra. In particular, we want to construct the quantum states $\ket{A\vec{b}}$ and $\ket{A^{-1}\vec{b}}$, corresponding to the matrix-vector product $A\vec{b}$ and the solution of the linear system $A\vec{x} = \vec{b}$:
\begin{theorem}[Quantum linear algebra with block encodings \cite{chakraborty2019power, gilyen2019quantum}] \label{qlsa} 
	Let $A \in \R^{n\times n}$ be a matrix with non-zero eigenvalues in the interval $[-1, -1/\kappa] \cup [1/\kappa, 1]$, and let $\epsilon > 0$. Given an implementation of an $(\zeta, O(\log n))$ block encoding for $A$ in time $T_{U}$ and a procedure for preparing state $\ket{b}$ in time $T_{b}$, 
	\begin{enumerate} 
		\item A state $\epsilon$-close to $\ket{A^{-1} b}$ can be generated in time 
		$O((T_{U} \kappa \zeta+ T_{b} \kappa) \polylog(\kappa \zeta /\epsilon))$. 
		\item A state $\epsilon$-close to $\ket{A b}$ can be generated in time $O((T_{U} \kappa \zeta+ T_{b} \kappa) \polylog(\kappa \zeta /\epsilon))$. 
		\item For  $\mathcal{A} \in \{ A, A^{-1} \}$, an estimate $\Lambda$ such that $\Lambda \in (1\pm \epsilon) \norm{ \mathcal{A} b}$ can be generated in time $O((T_{U}+T_{b}) \frac{ \kappa \zeta}{ \epsilon} \polylog(\kappa \zeta/\epsilon))$. 
	\end{enumerate}
\end{theorem}  
\noindent Finally, in order to recover classical information from the outputs of a linear system solver, we require an efficient procedure for \emph{quantum state tomography}. The tomography procedure 
is linear in the dimension of the quantum state. 
\begin{theorem}[Efficient vector state tomography, \cite{kerenidis2018quantum}]\label{vector state tomography}
	There exists an algorithm that given a procedure for constructing $\ket{\vec{x}}$ (i.e. a unitary mapping $U:\ket{0} \mapsto \ket{\vec{x}}$ \change{and its controlled version} in time $T_U$) and precision $\delta > 0$ produces an estimate $\measured{\vec{x}} \in \R^n$ with $\norm{\measured{\vec{x}}} = 1$ such that $\norm{\vec{x} - \measured{\vec{x}}} \leq \sqrt{7} \delta$ with probability at least $(1 - 1/n^{0.83})$. The algorithm runs in time $O\left( T_U \frac{n \log n}{\delta^2}\right)$.
\end{theorem}
Of course, repeating this algorithm $\widetilde{O}(1)$ times allows us to increase the success probability to at least $1 - 1/\poly(n)$. Putting Theorems \ref{qbe}, \ref{qlsa} and \ref{vector state tomography} together, assuming that $A$ and $\vec{b}$ are already in QRAM, we obtain that the complexity of a completely self-contained algorithm for solving the system $A\vec{x} = \vec{b}$ with error $\delta$ is $\widetilde{O} \left( n \cdot \frac{\kappa \zeta}{\delta^2} \right)$. For well-conditioned matrices, this presents a significant improvement over $O(n^\omega)$ (or, in practice, $O(n^3)$) needed for solving linear systems classically, especially when $n$ is large and the desired precision is not too high. This can be compared with classical iterative linear algebra algorithms \cite{saad2003iterative} that have $O(\operatorname{nnz}(A))$ complexity per iteration, as well as the new quantum-inspired solvers \cite{gilyen2018quantum} that have high-degree polynomial dependence on the rank, error, and the condition number.

\section{A quantum interior-point method}
Having introduced the classical IPM for SOCP, we can finally introduce our quantum IPM (Algorithm \ref{alg:qipm}). The main idea is to use quantum linear algebra as much as possible (including solving the Newton system -- the most expensive part of each iteration), and falling back to classical computation when needed. Since quantum linear algebra introduces inexactness, we need to deal with it in the analysis. The inspiration for the ``quantum part'' of the analysis is \cite{kerenidis2018quantum}, whereas the ``classical part'' is based on \cite{monteiro2000polynomial}. Nevertheless, the SOCP analysis is unique in many aspects and a number of hurdles had to be overcome to the make the analysis go through. In the rest of the section, we give a brief introduction of the most important quantum building blocks we use, as well as present a sketch of the analysis.

\begin{algorithm} 
	\caption{A quantum IPM for SOCP} \label{alg:qipm} 
		\textbf{Require:} Matrix $A$ and vectors $\vec{b}, \vec{c}$ in QRAM, precision parameter $\epsilon$\\
		\begin{enumerate} 
			\item Find feasible initial point $(\vec{x}, \vec{y}, \vec{s}, \mu)$ and store it in QRAM.
			\item Repeat the following steps for $O(\sqrt{r}\log(\mu_0/\epsilon))$ iterations:
			\begin{enumerate}
				\item Compute the vector $\sigma \mu \vec{e} - \vec{x} \circ \vec{s}$ classically and store it in QRAM. 
				\item Prepare and update the block encodings of the LHS and the RHS of the Newton system \eqref{eq:Newton system}						
				\item Solve the Newton system to obtain $\ket{(\dx;\dy;\ds)}$, and obtain a classical approximate solution $\measured{\left( \dx ; \dy ; \ds \right)}$ using tomography.
				\item Update $\vec{x} \gets \vec{x} + \dxm$, $\vec{s} \gets \vec{s} + \dsm$ and store in QRAM.
				\item Update $\mu \gets \frac1r \vec{x}^\top \vec{s}$. 
			\end{enumerate} 
			\item Output $(\vec{x}, \vec{y}, \vec{s})$.
		\end{enumerate}
\end{algorithm}

First, we note that the algorithms from the previous section allow us to ``forget'' that Algorithm~\ref{alg:qipm} is quantum, and treat it as a small modification of the classical IPM, where the system \eqref{eq:Newton system} is solved up to an $\ell_2$-error $\delta$. Since Algorithm~\ref{alg:qipm} is iterative, the main part of the analysis is proving that a single iteration preserves closeness to the central path, strict feasibility, and improves the duality gap. In the remainder of this section, we state our main results informally, while the exact statements and proofs of all claims can be found in the supplementary material.

\begin{theorem}[Per-iteration correctness, informal] \label{thm:main} 
	Let $(\vec{x}, \vec{y}, \vec{s})$ be a strictly feasible primal-dual solution that is close to the central path, with duality gap $\mu$, and at distance at least $\delta$ from the boundary of $\lorentz$. Then, the Newton system \eqref{eq:Newton system} has a unique solution $(\dx, \dy, \ds)$. There exist positive constants $\xi, \alpha$ such that the following holds: If we let $\dxm, \dsm$ be approximate solutions of \eqref{eq:Newton system} that satisfy
	\[
	\norm{\dx - \dxm}_F \leq \xi \delta \text{ and }
	\norm{\ds - \dsm}_F \leq \xi \delta,
	\] 
	and let $\xnext := \vec{x} + \dxm$ and $\snext := \vec{s} + \dsm$ be the updated solution, then:
	\begin{enumerate}
		\item The updated solution is strictly feasible, i.e. $\xnext \in \interior \lorentz$ and $\snext \in \interior \lorentz$.
		\item The updated solution is close to the central path, and the new duality gap is less than $(1-\alpha/\sqrt{r})\mu$.
	\end{enumerate}
\end{theorem}
The proof of this theorem consists of 3 main parts:
\begin{enumerate}
	\item Rescaling $\vec{x}$ and $\vec{s}$ so that they commute in the Jordan-algebraic sense \cite{alizadeh2003second}. This part can be reused from the classical analysis \cite{monteiro2000polynomial}.
	\item Bounding the norms of $\dxm$ and $\dsm$, and proving that $\vec{x} + \dxm$ and $\vec{s} + \dsm$ are still strictly feasible (in the sense of belonging to $\interior \lorentz$). This part of the analysis is also inspired by the classical analysis, but it has to take into account the inexactness of the Newton system solution.
	\item Proving that the new solution $(\vec{x} + \dxm, \vec{y} + \dym, \vec{s} + \dsm)$ is in the neighborhood of the central path, and the duality gap/central path parameter have decreased by a factor of $1 - \alpha/\sqrt{r}$, where $\alpha$ is constant. This part is the most technical, and while it is inspired by \cite{kerenidis2018quantum}, it required using many of the Jordan-algebraic tools from \cite{alizadeh2003second, monteiro2000polynomial}.
\end{enumerate}
Theorem~\ref{thm:main} formalizes the fact that the Algorithm~\ref{alg:qipm} has the same iteration invariant as the classical IPM. Since the duality gap is reduced by the same factor in both algorithms, their iteration complexity is the same, and a simple calculation shows that they need $O(\sqrt{r})$ iterations to halve the duality gap. On the other hand, the cost of each iteration varies, since the complexity of the quantum linear system solver depends on the precision $\xi\delta$, the condition number $\kappa$ of the Newton matrix, as well as its $\zeta$-parameter. While exactly bounding these quantities is the subject of future research, it is worth noting that for $\zeta$, we have the trivial bound $\zeta \leq \sqrt{n}$, and research on iterative linear algebra methods \cite{dollar2005iterative} suggests that $\kappa = O(1/\mu) = O(1/\epsilon)$. The final complexity is summarized in the following theorem:
\begin{theorem}\label{thm:runtime}
	Let \eqref{prob:SOCP primal} be a SOCP with $A \in \R^{m\times n}$, $m \leq n$, and $\lorentz = \lorentz^{n_1} \times \cdots \times \lorentz^{n_r}$. Then, Algorithm \ref{alg:qipm} achieves duality gap $\epsilon$ in time
	\begin{equation*}
	T = \widetilde{O} \left( \sqrt{r} \log\left( \mu_0 / \epsilon \right) \cdot \frac{n \kappa \zeta}{\delta^2}\log\left( \frac{\kappa \zeta}{\delta} \right)  \right),
	\end{equation*}
	where the $\widetilde{O}(\cdot)$ notation hides the factors that are poly-logarithmic in $n$ and $m$.
\end{theorem}
Finally, the quality of the resulting (classical) solution is characterized by the following theorem:
\begin{theorem}\label{thm:feasibility}
	Let \eqref{prob:SOCP primal} be a SOCP as in Theorem \ref{thm:runtime}. Then, after $T$ iterations, the (linear) infeasibility of the final iterate $\vec{x}, \vec{y}, \vec{s}$  is bounded as
	\begin{align*}
	\norm{A\vec{x} - \vec{b}} &\leq \delta\norm{A},\\
	\norm{A^\top \vec{y} + \vec{s} - \vec{c}} &\leq \delta \left( \norm{A} + 1 \right).
	\end{align*}
\end{theorem}

\section{Technical results}
In this section, we present our main technical results -- the proofs of Theorems~\ref{thm:main}, \ref{thm:runtime} and~\ref{thm:feasibility}.
\subsection{Central path}
In addition to the central path defined in \eqref{eq:central path}, we define the distance from the central path as $d(\vec{x}, \vec{s}, \nu) = \norm{T_{\vec{x}} \vec{s} - \nu \vec{e}}_F$, so the corresponding $\eta$-neighborhood is given by 
\begin{align*}
	\mathcal{N}_\eta(\nu) := \{ &(\vec{x}, \vec{y}, \vec{s})\;|\;(\vec{x}, \vec{y}, \vec{s}) \;\text{strictly feasible} \text{ and } d(\vec{x}, \vec{s}, \nu) \leq \eta \nu \}.
\end{align*}
Using this neighborhood definition, we can specify what exactly do we mean when we claim that the important properties of the central path are valid in its neighborhood as well.
\begin{lemma}[Properties of the central path]\label{lemma:Properties of the central path}
	Let $\nu > 0$ be arbitrary and let $\vec{x}, \vec{s} \in \interior\lorentz$. Then, $\vec{x}$ and $\vec{s}$ satisfy the following properties:
	\begin{enumerate}
		\item For all $\nu > 0$, the duality gap and distance from the central path are related as 
		\begin{align*}
		    | \vec{x}^\top \vec{s} - r\nu | \leq \sqrt{\frac{r}{2}} \cdot d(\vec{x}, \vec{s}, \nu).
		\end{align*}
		\item The distance from the central path is symmetric in its arguments i.e. $d(\vec{x}, \vec{s}, \nu) = d(\vec{s}, \vec{x}, \nu)$.
		\item Let $\mu = \frac1r \vec{x}^\top \vec{s}$. If $d(\vec{x}, \vec{s}, \mu) \leq \eta \mu$, then $(1+\eta) \norm{\vec{s}^{-1}}_2 \geq \norm{\mu^{-1} \vec{x}}_2$. 
	\end{enumerate}
\end{lemma}
\begin{proof}
	For part 1, let $\{\lambda_i\}_{i=1}^{2r}$ be the eigenvalues of $T_{\vec{x}}\vec{s}$, note that $T_{x}$ is invertible as $x \in \lorentz$. 
	Then using the properties of $T_{\vec{x}}$, we have
	\begin{align*}
	\vec{x}^\top\vec{s} &= \vec{x}^\top T_{\vec{x}}^{-1} T_{\vec{x}} \vec{s} =  (T_{\vec{x}^{-1}} \vec{x})^\top T_{\vec{x}} \vec{s} =\vec{e}^\top T_{\vec{x}} \vec{s} = \frac12 \sum_{i=1}^{2r} \lambda_i.
	\end{align*}
	We can therefore bound the duality gap $x^\top s$ as follows, 
	\begin{align*}
		\vec{x}^\top \vec{s} = \frac12 \sum_{i=1}^{2r} \lambda_i &\leq r\nu + \frac12 \sum_{i=1}^{2r} |\lambda_i - \nu| \leq r\nu + \sqrt{\frac{r}{2}}\sqrt{\sum_{i=1}^{2r} (\lambda_i - \nu)^2} = r\nu + \sqrt{\frac{r}{2}} \cdot d(\vec{x}, \vec{s}, \nu).
	\end{align*}
	The second step used the Cauchy-Schwarz inequality while the third follows from the definition $d(\vec{x}, \vec{s}, \nu)^2 = \sum_{i=1}^{2r} (\lambda_i - \nu)^2$. 
	The proof of the lower bound is similar, but starts instead with the inequality
	\[
	\frac12\sum_{i=1}^{2r} \lambda_i \geq r\nu - \frac12 \sum_{i=1}^{2r} |\nu - \lambda_i|.
	\]
	For part 2, it suffices to prove that $T_{\vec{x}} \vec{s}$ and $T_{\vec{s}} \vec{x}$ have the same eigenvalues. This follows from part 2 of Theorem 10 in \cite{alizadeh2003second}.
	Finally for part 3, as $d(\vec{s}, \vec{x}, \mu) \leq \eta \mu$ we have, 
	\begin{align*}
		\eta \mu &\geq \norm{T_{\vec{s}}\vec{x} - \mu \vec{e}}_{F} \\
		&= \norm{T_{\vec{s}} \vec{x} - \mu \left( T_{\vec{s}} T_{\vec{s}^{-1}} \right) \vec{e} }_{F} \\
		&= \norm{T_{\vec{s}} \left( \vec{x} - \mu T_{\vec{s}^{-1}}\vec{e} \right) }_{F} \\
		&\geq \lambda_\text{min} (T_{\vec{s}}) \norm{\vec{x} - \mu T_{\vec{s}^{-1}}\vec{e}}_{F} \\
		&\geq \lambda_\text{min} (T_{\vec{s}}) \norm{\vec{x} - \mu \vec{s}^{-1}}_{2} \\
		&= \frac{1}{\norm{\vec{s}^{-1}}_{2}} \cdot \mu \cdot \norm{ \mu^{-1} \vec{x} - \vec{s}^{-1} }_{2}
	\end{align*}
	Therefore, $\eta \norm{\vec{s}^{-1}}_{2} \geq \norm{\mu^{-1} \vec{x} - \vec{s}^{-1}}_{2}$. Finally, by the triangle inequality for the spectral norm,  
	\[
	\eta \norm{\vec{s}^{-1}}_{2} \geq \norm{\mu^{-1} \vec{x}}_{2} -\norm{\vec{s}^{-1}}_{2},
	\]
	so we can conclude that $\norm{\vec{s}^{-1}}_{2} \geq \frac{1}{1+\eta} \norm{\mu^{-1} \vec{x}}_{2}$.
\end{proof}

\subsection{A single quantum IPM iteration}
Recall that the essence of our quantum algorithm is repeated solution of the Newton system \eqref{eq:Newton system} using quantum linear algebra. As such, our goal is to prove the following theorem:
\setcounter{theorem}{3}
\begin{theorem}[Per-iteration correctness, formal]
	Let $\chi = \eta = 0.01$ and $\xi = 0.001$ be positive constants and let $(\vec{x}, \vec{y}, \vec{s})$ be a feasible solution of \eqref{prob:SOCP primal} and \eqref{prob:SOCP dual} with $\mu = \frac1r \vec{x}^\top \vec{s}$ and $d(\vec{x}, \vec{s}, \mu) \leq \eta\mu$. Then, for $\sigma = 1-\chi/\sqrt{n}$, the Newton system \eqref{eq:Newton system} has a unique solution $(\dx, \dy, \ds)$. Let $\dxm, \dsm$ be approximate solutions of \eqref{eq:Newton system} that satisfy
	\[
	\norm{\dx - \dxm}_F \leq \frac{\xi}{\norm{T_{\vec{x}^{-1}}}} ,\;
	\norm{\ds - \dsm}_F \leq \frac{\xi}{2\norm{T_{\vec{s}^{-1}}}},
	\]
	where $T_{\vec{x}}$ and $T_{\vec{x}}$ are the square roots of the quadratic representation matrices in equation \eqref{qrep}. 
	If we let $\xnext := \vec{x} + \dxm$ and $\snext := \vec{s} + \dsm$, the following holds:
	\begin{enumerate}
		\item The updated solution is strictly feasible, i.e. $\xnext \in \interior \lorentz$ and $\snext \in \interior \lorentz$.
		\item The updated solution satisfies $d(\xnext, \snext, \measured{\mu}) \leq \eta\measured{\mu}$ and $\frac1r \xnext^\top \snext = \measured{\mu}$ for $\measured{\mu} = \measured{\sigma}\mu$, $\measured{\sigma} = 1 - \frac{\alpha}{\sqrt{r}}$ and a constant $0 < \alpha \leq \chi$.
	\end{enumerate}
\end{theorem}
Since the Newton system \eqref{eq:Newton system} is the same as in the classical case, we can reuse Theorem 1 from \cite{monteiro2000polynomial} for the uniqueness part of Theorem \ref{thm:main}. Therefore, we just need to prove the two parts about strict feasibility and improving the duality gap.
Our analysis is inspired by the general case analysis from \cite{ben2001lectures}, the derived SDP analysis from \cite{kerenidis2018quantum}, and uses some technical results from the SOCP analysis in \cite{monteiro2000polynomial}.
The proof of Theorem \ref{thm:main} consists of three main steps:
\begin{enumerate}
	\item Rescaling $\vec{x}$ and $\vec{s}$ so that they share the same Jordan frame.
	\item Bounding the norms of $\dx$ and $\ds$, and proving that $\vec{x} + \dx$ and $\vec{s} + \ds$ are still strictly feasible (in the sense of belonging to $\interior \lorentz$).
	\item Proving that the new solution $(\vec{x} + \dx, \vec{y} + \dy, \vec{s} + \ds)$ is in the $\eta$-neighborhood of the central path, and the duality gap/central path parameter have decreased by a factor of $1 - \alpha/\sqrt{n}$, where $\alpha$ is constant.
\end{enumerate}

\subsection{Rescaling \texorpdfstring{$\vec{x}$}{x} and \texorpdfstring{$\vec{s}$}{s}}
As in the case of SDPs, the first step of the proof uses the symmetries of the Lorentz cone to perform a commutative scaling, that is to reduce the analysis to the case when $\vec{x}$ and $\vec{s}$ share the same Jordan frame. Although $\circ$ is commutative by definition, two vectors sharing a Jordan frame are akin to two matrices sharing a system of eigenvectors, and thus commuting (some authors \cite{alizadeh2003second} say that the vectors \emph{operator commute} in this case).
The easiest way to achieve this is to scale by $T_{\vec{x}} = Q_{\vec{x}^{1/2}}$ and $\mu^{-1}$, i.e. to change our variables as
\[
\vec{x} \mapsto \scaled{\vec{x}} := T_{\vec{x}}^{-1} \vec{x} = \vec{e} \text{ and } \vec{s} \mapsto \scaled{\vec{s}} := \mu^{-1}T_{\vec{x}} \vec{s}.
\]
Note that for convenience, we have also rescaled the duality gap to 1. Recall also that in the matrix case, the equivalent of this scaling was $X \mapsto X^{-1/2} X X^{-1/2} = I$ and $S \mapsto \mu^{-1}X^{1/2} S X^{1/2}$.
We use the notation $\scaled{\vec{z}}$ to denote the appropriately-scaled vector $\vec{z}$, so that we have
\[
\dxs := T_{\vec{x}}^{-1} \dx,\quad \dss := \mu^{-1}T_{\vec{x}} \ds
\]
For approximate quantities (e.g. the ones obtained using tomography, or any other approximate linear system solver), we use the notation $\measured{\phantom{i}\cdot\phantom{i}}$, so that the increments become $\dxm$ and $\dsm$, and their scaled counterparts are $\dxms := T_{\vec{x}}^{-1} \dxm$ and $\dsms := \mu^{-1}T_{\vec{x}} \dsm$. Finally, we denote the scaled version of the next iterate as $\xsnext := \vec{e} + \dxms$ and 
$\ssnext := \ss + \dsms$. Now, we see that the statement of Theorem \ref{thm:main} implies the following bounds on $\norm{ \dxs - \dxms }_F$ and $\norm{\dss - \dsms}_F$:
\begin{align*}
	\norm{\dxs - \dxms}_F & = \norm{T_{\vec{x}^{-1}}\dx - T_{\vec{x}^{-1}} \dxm}_F \\
	&\leq \norm{T_{\vec{x}^{-1}}} \cdot \norm{\dx - \dxm}_F \leq \xi, \text{ and } \\
	\norm{\dss - \dsms}_F &= \mu^{-1}\norm{T_{\vec{x}} \ds - T_{\vec{x}} \dsm}_F \\
	&\leq \mu^{-1}\norm{T_{\vec{x}}} \norm{\ds - \dsm}_F \\
	&= \mu^{-1} \norm{\vec{x}}_2 \norm{\ds - \dsm}_F \\
	&\leq (1+\eta) \norm{\vec{s}^{-1}}_2 \norm{\ds - \dsm}_F \text{ by Lemma \ref{lemma:Properties of the central path}} \\
	&\leq 2 \norm{T_{\vec{s}^{-1}}} \norm{\ds - \dsm}_F \leq \xi. 
\end{align*}

Throughout the analysis, we will make use of several constants: $\eta > 0$ is the distance from the central path, i.e. we ensure that our iterates stay in the $\eta$-neighborhood $\mathcal{N}_\eta$ of the central path. The constant $\sigma = 1 - \chi / \sqrt{r}$ is the factor by which we aim to decrease our duality gap, for some constant $\chi > 0$. Finally constant $\xi > 0$ is the approximation error for the scaled increments $\dxms, \dsms$.
Having this notation in mind, we can state several facts about the relation between the duality gap and the central path distance for the original and scaled vectors.
\begin{claim}\label{claim:stuff preserved under scaling}
	The following holds for the scaled vectors $\xs$ and $\ss$:
	\begin{enumerate}
		\item The scaled duality gap is $\frac1r \xs^\top \ss = 1$.
		\item $d(\vec{x}, \vec{s}, \mu) \leq \eta \mu$ is equivalent to $\norm{\ss - \vec{e}} \leq \eta$.
		\item $d(\vec{x}, \vec{s}, \mu\sigma) = \mu \cdot d(\xs, \ss, \sigma)$, for all $\sigma > 0$.
	\end{enumerate}
\end{claim}
At this point, we claim that it suffices to prove the two parts of Theorem \ref{thm:main} in the scaled case. Namely, assuming that $\xsnext \in \interior \lorentz$ and $\ssnext \in \interior\lorentz$, by construction we get
\begin{equation*}
	\xnext = T_{\vec{x}} \xsnext \text{ and } \snext = T_{\vec{x}^{-1}} \ssnext
\end{equation*}
and thus $\xnext, \snext \in \interior \lorentz$. 

On the other hand, if $\mu d(\xsnext, \ssnext, \measured{\sigma}) \leq \eta \measured{\mu}$, then $d(\xnext, \snext, \measured{\mu}) \leq \eta \measured{\mu}$ follows by Claim \ref{claim:stuff preserved under scaling}. Similarly, from $\frac1r \xsnext^\top \ssnext = \measured{\sigma}$, we also get $\frac1r \vec{x}_\text{next}^\top \vec{s}_\text{next} = \measured{\mu}$. 
We conclude this part with two technical results from \cite{monteiro2000polynomial}, that use the auxiliary matrix $R_{xs}$ defined as $R_{xs} := T_{\vec{x}} \Arw(\vec{x})^{-1} \Arw(\vec{s}) T_{\vec{x}}$. These results are useful for the later parts of the proof of Theorem \ref{thm:main}.
\begin{claim}[\cite{monteiro2000polynomial}, Lemma 3]\label{claim:R_xs bound}
	Let $\eta$ be the distance from the central path, and let $\nu > 0$ be arbitrary. Then, $R_{xs}$ is bounded as
	\[
	\norm{R_{xs} - \nu I} \leq 3 \eta \nu.
	\]
\end{claim}
\begin{claim}[\cite{monteiro2000polynomial}, Lemma 5, proof]\label{claim:dss expression}
	Let $\mu$ be the duality gap. Then, the scaled increment $\dss$ is
	\[
	\dss = \sigma \vec{e} - \ss - \mu^{-1}R_{xs} \dxs.
	\]
\end{claim}	

\subsection{Maintaining strict feasibility}
The main tool for showing that strict feasibility is conserved is the following bound on the increments $\dxs$ and $\dss$:
\begin{lemma}[\cite{monteiro2000polynomial}, Lemma 6]\label{lemma:bounds for increment}
	Let $\eta$ be the distance from the central path and let $\mu$ be the duality gap. Then, we have the following bounds for the scaled direction:
	\[
	\begin{array}{rl}
		\norm{\dxs}_F &\leq \frac{\Theta}{\sqrt{2}} \\
		\norm{\dss}_F &\leq \Theta \sqrt{2}
	\end{array}
	, \quad\text{where}\quad \Theta = \frac{2\sqrt{\eta^2 / 2+ (1-\sigma)^2 r}}{1 - 3\eta}
	\]
\end{lemma}
\noindent Moreover, if we substitute $\sigma$ with its actual value $1 - \chi/\sqrt{r}$, we get $\Theta = \frac{\sqrt{2\eta^2 + 4\chi^2}}{1-3\eta}$, which we can make arbitrarily small by tuning the constants. Now, we can immediately use this result to prove $\xsnext, \ssnext \in \interior \lorentz$.
\begin{lemma} \label{l1} 
	Let $\eta = \chi = 0.01$ and $\xi = 0.001$. Then, $\xsnext$ and $\ssnext$ are strictly feasible, i.e. $\xsnext, \ssnext \in \interior \lorentz$.
\end{lemma}
\begin{proof} 
	By Lemma \ref{lemma:bounds for increment}, $\lambda_\text{min}(\xm) \geq 1 - \norm{\dxms}_F \geq 1 - \frac{\Theta}{\sqrt{2}} - \xi$. On the other hand, since $d(\vec{x}, \vec{s}, \mu) \leq \eta \mu$, we have $d(\xs, \ss, 1) \leq \eta$, and thus
	\begin{align*}
		\eta^2  &\geq \norm{\ss - e}_F^2 = \sum_{i=1}^{2r} (\lambda_i(\ss)-1)^2
	\end{align*}
	The above equation implies that $\lambda_i(\ss) \in \left[ 1-\eta, 1+\eta \right], \forall i \in [2r]$.
	Now, since $\norm{\vec{z}}_2\leq \norm{\vec{z}}_F$, 
	\begin{align*}
		\lambda_\text{min}(\measured{\vec{s}}) &\geq \lambda_\text{min}(\ss + \dss) - \norm{\dsms - \dss}_F \\
		&\geq \lambda_\text{min}(\ss) - \norm{\dss}_F - \norm{\dsms - \dss}_F \\
		&\geq 1-\eta - \Theta \sqrt{2} - \xi,
	\end{align*}
	where we used Lemma \ref{lemma:bounds for increment} for the last inequality.
	Substituting $\eta = \chi = 0.01$ and $\xi = 0.001$, we get that $\lambda_\text{min}(\xm) \geq 0.8$ and $\lambda_\text{min}(\measured{\vec{s}}) \geq 0.8$.
\end{proof} 

\subsection{Maintaining closeness to central path}
Finally, we move on to the most technical part of the proof of Theorem \ref{thm:main}, where we prove that $\xsnext, \ssnext$ is still close to the central path, and the duality gap has decreased by a constant factor. We split this into two lemmas.

\begin{lemma}\label{lemma:we stay on central path}
	Let $\eta = \chi = 0.01$, $\xi=0.001$, and let $\alpha$ be any value satisfying $0 < \alpha \leq \chi$. Then, for $\measured{\sigma} = 1 - \alpha / \sqrt{r}$, the distance to the central path is maintained, that is, $d(\xsnext, \ssnext, \measured{\sigma}) < \eta\measured{\sigma}$.
\end{lemma}
\begin{proof} 
	By Claim \ref{claim:stuff preserved under scaling}, the distance of the next iterate from the central path is \[d(\xsnext, \ssnext, \measured{\sigma}) = \norm{T_{\xsnext} \ssnext - \measured{\sigma}\vec{e}}_F,\] and we can bound it from above as
	\begin{align*}
		d(\xsnext, \ssnext, \measured{\sigma}) &= \norm{T_{\xsnext} \ssnext - \measured{\sigma}\vec{e}}_F \\
		&= \norm{T_{\xsnext} \ssnext - \measured{\sigma} T_{\xsnext}T_{\xsnext^{-1}} \vec{e}}_F \\
		&\leq \norm{T_{\xsnext}} \cdot \norm{\ssnext - \measured{\sigma}\cdot (\xsnext)^{-1}}.
	\end{align*}
	So, it is enough to bound $\norm{\vec{z}}_F := \norm{\ssnext - \measured{\sigma}\cdot \xsnext^{-1}}_F$ from above, since
	\begin{align*}
	\norm{T_{\xsnext}} &= \norm{\xsnext}_2 \leq 1 + \norm{\dxms}_2 \leq 1 + \norm{\dxs}_2 + \xi \leq 1 + \frac{\Theta}{\sqrt{2}} + \xi.
	\end{align*}
	We split $\vec{z}$ as
	\begin{align*}
		\vec{z} &= \underbrace{\left( \ss + \dsms - \measured{\sigma}e + \dxms \right)}_{\vec{z}_1} 
		+ \underbrace{(\measured{\sigma} - 1)\dxms}_{\vec{z}_2} + \underbrace{\measured{\sigma} \left( e - \dxms - (e + \dxms)^{-1} \right) }_{\vec{z}_3},
	\end{align*}
	and we bound $\norm{\vec{z}_1}_F, \norm{\vec{z}_2}_F$, and $\norm{\vec{z}_3}_F$ separately.
	
	\begin{enumerate}
		\item By the triangle inequality, $\norm{\vec{z}_1}_F \leq \norm{\ss + \dss - \measured{\sigma} e + \dxs}_F + 2\xi$. Furthermore, after substituting $\dss$ from Claim \ref{claim:dss expression}, 
		we get 
		\begin{align*}
			\ss + \dss - \measured{\sigma} e + \dxs &= \sigma e - \mu^{-1}R_{xs} \dxs - \measured{\sigma} e + \dxs \\
			&= \frac{\alpha - \chi}{\sqrt{r}} e + \mu^{-1}(\mu I - R_{xs})\dxs.
		\end{align*}
		Using the bound for $\norm{\mu I - R_{xs}}$ from Claim \ref{claim:R_xs bound} as well as the bound for $\norm{\dxs}_F$ from Lemma \ref{lemma:bounds for increment}, we obtain
		\begin{align*}
			\norm{\vec{z}_1}_F \leq 2\xi + \frac{\chi}{\sqrt{r}} + \frac{3}{\sqrt{2}}\eta\Theta.
		\end{align*}
		
		\item $\norm{\vec{z}_2}_F \leq \frac{\chi}{\sqrt{r}} \left( \frac{\Theta}{\sqrt{2}} + \xi \right)$, where we used the bound from Lemma \ref{lemma:bounds for increment} again.
		
		\item Here, we first need to bound $\norm{ (e + \dxs)^{-1} - (e + \dxms)^{-1} }_F$. For this, we use the submultiplicativity of $\norm{\cdot}_F$:
		\begin{align*}
			\norm{& (e + \dxs)^{-1} - (e + \dxms)^{-1} }_F = \norm{ (e+\dxs)^{-1} \circ \left( e - (e+\dxs) \circ (e + \dxms)^{-1} \right) }_F\\
			&\leq \norm{ (e+\dxs)^{-1} }_2 \cdot \norm{ e - (e+\dxms + \dxs - \dxms) \circ (e + \dxms)^{-1} }_F \\
			&= \norm{ (e+\dxs)^{-1} }_2 \cdot \norm{ (\dxs - \dxms) \circ (e + \dxms)^{-1} }_F \\
			&\leq \norm{ (e+\dxs)^{-1} }_2 \cdot \norm{ \dxs - \dxms }_F \cdot \norm{ (e + \dxms)^{-1} }_2 \\
			&\leq \xi \cdot \norm{ (e+\dxs)^{-1} }_2 \cdot \norm{ (e + \dxms)^{-1} }_2.
		\end{align*}
		Now, we have the bound $\norm{(e+\dxs)^{-1}}_2 \leq \frac{1}{1 - \norm{\dxs}_F}$ and similarly $\norm{(e+\dxms)^{-1}}_2 \leq \frac{1}{1 - \norm{\dxs}_F - \xi}$, so we get
		\[
		\norm{ (e + \dxs)^{-1} - (e + \dxms)^{-1} }_F \leq \frac{\xi}{(1 - \norm{\dxs}_F - \xi)^2}.
		\]
		Using this, we can bound $\norm{ \vec{z}_3 }_F$:
		\begin{align*}
			\norm{ \vec{z}_3 }_F \leq \measured{\sigma}&\left( \norm{ e - \dxs - (e+\dxs)^{-1} }_F + \xi + \frac{\xi}{(1 - \norm{\dxs}_F - \xi)^2} \right).
		\end{align*}
		If we let $\lambda_i$ be the eigenvalues of $\dxs$, then by Lemma \ref{lemma:bounds for increment}, we have
		\begin{align*}
			\norm{& e - \dxs - (e+\dxs)^{-1} }_F \\
			&= \sqrt{ \sum_{i=1}^{2r} \left( (1-\lambda_{i})-\frac{1}{1+\lambda_{i}} \right)^2 } \\
			&= \sqrt{\sum_{i=1}^{2r} \frac{\lambda_{i}^4}{(1+\lambda_{i})^2}} 
			\leq \frac{\Theta}{\sqrt2 - \Theta} \sqrt{\sum_{i=1}^{2r} \lambda_{i}^2} \\
			&\leq \frac{\Theta^2}{2 - \sqrt{2}\Theta}.
		\end{align*}
	\end{enumerate}
	Combining all bound from above, we obtain
	\begin{align*}
		d(\xsnext, \ssnext, \measured{\sigma}) \leq \left( 1 + \frac{\Theta}{\sqrt{2}} + \xi \right) \cdot 
		&\left(2\xi + \frac{\chi}{\sqrt{r}} + \frac{3}{\sqrt{2}}\eta\Theta\right. \\
		&+\frac{\chi}{\sqrt{r}} \left( \frac{\Theta}{\sqrt{2}} + \xi \right) \\
		&+\left. \measured{\sigma}\left( \frac{\Theta^2}{2 - \sqrt{2}\Theta} + \xi + \frac{\xi}{(1 - \Theta/\sqrt{2} - \xi)^2} \right) \right).
	\end{align*}
	Finally, if we plug in $\chi=0.01$, $\eta=0.01$, $\xi= 0.001$, we get $d(\xm, \sm, \measured{\sigma}) \leq 0.005\measured{\sigma} \leq \eta\measured{\sigma}$.
\end{proof} 
Now, we prove that the duality gap decreases.
\begin{lemma} \label{gap} 
	For the same constants, the updated solution satisfies $\frac1r \xsnext^\top \ssnext = \left( 1-\frac{\alpha}{\sqrt{r}} \right)$ for $\alpha = 0.005$.
\end{lemma}
\begin{proof}
	Since $\xsnext$ and $\ssnext$ are scaled quantities, the duality gap between their unscaled counterparts is $\frac{\mu}{r}\xsnext^\top \ssnext$. Applying Lemma \ref{lemma:Properties of the central path} (and Claim \ref{claim:stuff preserved under scaling}) with $\nu=\sigma\mu$ and $\xsnext, \ssnext$, we obtain the upper bound
	\begin{align*}
		\mu\xsnext^\top \ssnext \leq r \sigma\mu + \sqrt{\frac{r}{2}} \mu d(\xsnext, \ssnext, \sigma),
	\end{align*}
	which in turn implies
	\begin{equation*}
		\frac1r \xsnext^\top \ssnext \leq \left( 1 - \frac{0.01}{\sqrt{r}} \right) \left( 1 + \frac{d(\xsnext, \ssnext, \sigma)}{\sigma \sqrt{2r}} \right).
	\end{equation*}
	By instantiating Lemma \ref{lemma:Properties of the central path} for $\alpha = \chi$, from its proof, we obtain $d(\xsnext, \ssnext, \sigma) \leq 0.005 \sigma$, and thus 
	\begin{equation*}
		\frac1r \xsnext^\top \ssnext \leq 1 - \frac{0.005}{\sqrt{r}}
	\end{equation*}
	Therefore, the final $\alpha$ for this Lemma is $0.005$.
\end{proof} 

\subsection{Final complexity and feasibility}
In every iteration we need to solve the Newton system to a precision dependent on the norms of $T_{\vec{x}^{-1}}$ and $T_{\vec{s}^{-1}}$. Thus, to bound the running time of the algorithm (since the complexity of the vector state tomography procedure depends on the desired precision), we need to bound $\norm{T_{\vec{x}^{-1}}}$ and $\norm{T_{\vec{s}^{-1}}}$. Indeed, by the properties of the quadratic representation, we get 
\begin{align*}
	\norm{T_{\vec{x}^{-1}}} &= \norm{\vec{x}^{-1}} = \lambda_\text{min}(\vec{x})^{-1} \text{ and }\\ \norm{T_{\vec{s}^{-1}}} &= \norm{\vec{s}^{-1}} = \lambda_\text{min}(\vec{s})^{-1}.
\end{align*}
If the tomography precision for iteration $i$ is chosen to be at least (i.e. smaller than)
\begin{equation}\label{eq:definition of delta}
	\delta_i := \frac{\xi}{4} \min \left\{ \lambda_\text{min}(\vec{x}_i), \lambda_\text{min}(\vec{s}_i) \right\},
\end{equation}
then the premises of Theorem \ref{thm:main} are satisfied. The tomography precision for the entire algorithm can therefore be chosen to be $\delta := \min_i \delta_i$. Note that these minimum eigenvalues are related to how close the current iterate is to the boundary of $\lorentz$ -- as long as $\vec{x}_i, \vec{s}_i$ are not ``too close'' to the boundary of $\lorentz$, their minimal eigenvalues should not be ``too small''.

There are two more parameters that impact the complexity of the quantum linear system solver: the condition number of the Newton matrix $\kappa_i$ and the matrix parameter $\zeta_i$ of the QRAM encoding in iteration $i$. For both of these quantities we define their global versions as $\kappa = \max_i \kappa_i$ and $\zeta = \max_i \zeta_i$. Therefore, we arrive to the following statement about the complexity of Algorithm~\ref{alg:qipm}.
\begin{theorem}[restatement]
	Let \eqref{prob:SOCP primal} be a SOCP with $A \in \R^{m\times n}$, $m \leq n$, and $\lorentz = \lorentz^{n_1} \times \cdots \times \lorentz^{n_r}$. Then, our algorithm achieves duality gap $\epsilon$ in time
	\begin{equation*}
		T = \widetilde{O} \left( \sqrt{r} \log\left( \mu_0 / \epsilon \right) \cdot \frac{n \kappa \zeta}{\delta^2}\log\left( \frac{\kappa \zeta}{\delta} \right)  \right).
	\end{equation*}
\end{theorem}
This complexity can be easily interpreted as product of the number of iterations and the cost of $n$-dimensional vector tomography with error $\delta$. So, improving the complexity of the tomography algorithm would improve the running time of our algorithm as well.

Note that up to now, we cared mostly about strict (conic) feasibility of $\vec{x}$ and $\vec{s}$. Now, we address the fact that the linear constraints $A\vec{x} = \vec{b}$ and $A^\top \vec{y} + \vec{s} = \vec{c}$ are not exactly satisfied during the execution of the algorithm. Luckily, it turns out that this error is not accumulated, but is instead determined just by the final tomography precision:
\begin{theorem}[restatement]
	Let \eqref{prob:SOCP primal} be a SOCP as in Theorem \ref{thm:runtime}. Then, after $T$ iterations, the (linear) infeasibility of the final iterate $\vec{x}, \vec{y}, \vec{s}$  is bounded as
	\begin{align*}
		\norm{A\vec{x}_T - \vec{b}} &\leq \delta\norm{A}  , \\
		\norm{A^\top \vec{y}_T + \vec{s}_T - \vec{c}} &\leq \delta \left( \norm{A} + 1 \right).
	\end{align*}
\end{theorem}
\begin{proof}
	Let $(\vec{x}_T, \vec{y}_T, \vec{s}_T)$ be the $T$-th iterate. Then, the following holds for $A\vec{x}_T - \vec{b}$:
	\begin{equation}\label{eq:eq11}
		A \vec{x}_T - \vec{b} = A\vec{x}_0 + A\sum_{t=1}^\top \dxm_t - \vec{b} = A \sum_{t=1}^\top \dxm_t.
	\end{equation}
	On the other hand, the Newton system at iteration $T$ has the constraint $A\dx_T = \vec{b} - A\vec{x}_{T-1}$, which we can further recursively transform as, 
	\begin{align*}
		A\dx_T &= \vec{b} - A\vec{x}_{T-1} = \vec{b} - A\left( \vec{x}_{T-2} + \dxm_{T-1} \right) \\
		&= \vec{b} - A\vec{x}_0 - \sum_{t=1}^{T-1} \dxm_t = - \sum_{t=1}^{T-1} \dxm_t.
	\end{align*}
	Substituting this into equation \eqref{eq:eq11}, we get
	\[
	A\vec{x}_T - \vec{b} = A \left( \dxm_T - \dx_T \right).
	\]
	Similarly, using the constraint $A^\top\dy_T  + \ds_{T} = \vec{c} - \vec{s}_{T-1} - A^\top \vec{y}_{T-1}$ we obtain that
	\[
	A^\top \vec{y}_T + \vec{s}_T - \vec{c} = A^\top \left( \dym_T - \dy_T \right) + \left( \dsm_T - \ds_T \right).
	\]
	Finally, we can bound the norms of these two quantities, 
	\begin{align*}
		\norm{A\vec{x}_T - \vec{b}} &\leq \delta\norm{A} , \\
		\norm{A^\top \vec{y}_T + \vec{s}_T - \vec{c}} &\leq \delta \left( \norm{A} +1  \right).
	\end{align*}
\end{proof}

\section{Quantum Support-Vector Machines}
In this section we present our quantum support vector machine (SVM) algorithm as an application of our SOCP solver. 
Given a set of vectors $\mathcal{X} = \{ \vec{x}^{(i)} \in \R^n\;|\;i \in [m] \}$ (\emph{training examples}) and their \emph{labels} $y^{(i)} \in \{-1, 1\}$, the objective of the SVM training process is to find the ``best'' hyperplane that separates training examples with label $1$ from those with label $-1$. In this paper we focus on the (traditional) \emph{soft-margin} ($\ell_1$-)SVM, which can be expressed as the following optimization problem:
\begin{equation}
	\begin{array}{ll}
	\min\limits_{\vec{w}, b, \vec{\xi}} & \norm{\vec{w}}^2 + C\norm{\vec{\xi}}_1 \\
	\text{s.t.}& y^{(i)}(\vec{w}^\top\vec{x}^{(i)}+b) \geq 1 - \xi_i, \;\forall i \in [m] \\
	&\vec{\xi} \geq 0.
	\end{array} \label{prob:SVM}
\end{equation} 
Here, the variables $\vec{w}\in\R^n$ and $b \in \R$ correspond to the hyperplane, $\vec{\xi} \in \R^m$ corresponds to the ``linear inseparability'' of each point, and the constant $C > 0$ is a hyperparameter that quantifies the tradeoff between maximizing the margin and minimizing the constraint violations.

As a slightly less traditional alternative, one might also consider the $\ell_2$-SVM (or least-squares SVM, LS-SVM) \cite{suykens1999least}, where the $\norm{\vec{\xi}}_1$ regularization term is replaced by $\norm{\vec{\xi}}^2$. This formulation arises from considering the least-squares regression problem with the constraints $y^{(i)}(\vec{w}^\top\vec{x}^{(i)}+b) = 1$, which we solve by minimizing the squared $2$-norm of the residuals:
\begin{equation}
\begin{array}{ll}
\min\limits_{\vec{w}, b, \vec{\xi}} & \norm{\vec{w}}^2 + C \norm{\vec{\xi}}^2 \\
\text{s.t.}& y^{(i)}(\vec{w}^\top\vec{x}^{(i)}+b) = 1 - \xi_i, \;\forall i \in [m]
\end{array} \label{prob:LS-SVM}
\end{equation}
Since this is a least-squares problem, the optimal $\vec{w}, b$ and $\vec{\xi}$ can be obtained by solving a linear system. In \cite{rebentrost2014quantum}, a quantum algorithm for LS-SVM is presented, which uses a single quantum linear system solver. Unfortunately, replacing the $\ell_1$-norm with $\ell_2$ in the objective of \eqref{prob:LS-SVM} leads to the loss of a key property of ($\ell_1$-)SVM -- weight sparsity \cite{suykens2002weighted}.

\subsection{Reducing SVM to SOCP}
Finally, we are going to reduce the SVM problem \eqref{prob:SVM} to SOCP. In order to do that, we define an auxiliary vector $\vec{t} = \left( t+1; t; \vec{w} \right)$, where $t \in \R$ -- this allows us to ``compute'' $\norm{\vec{w}}^2$ using the constraint $\vec{t} \in \lorentz^{n+2}$ since
\begin{equation*}
\vec{t} \in \lorentz^{n+2} \Leftrightarrow (t+1)^2 \geq t^2 + \norm{\vec{w}}^2 \Leftrightarrow 2t+1 \geq \norm{\vec{w}}^2.
\end{equation*}
Thus, minimizing $\norm{\vec{w}}^2$ is equivalent to minimizing $t$. Note we can restrict our bias $b$ to be nonnegative without any loss in generality, since the case $b < 0$ can be equivalently described by a bias $-b > 0$ and weights $-\vec{w}$. Using these transformations, we can restate \eqref{prob:SVM} as the following SOCP:
\begin{equation}
\begin{array}{ll}
\min\limits_{\vec{t}, b, \vec{\xi}} & \begin{bmatrix}
0 & 1 & 0^n & 0 & C^m
\end{bmatrix} \begin{bmatrix}
\vec{t} & b & \vec{\xi}
\end{bmatrix}^\top \\
\text{s.t.}& 
\begin{bmatrix}
0 & 0 & & 1 & \\
\vdots & \vdots& X^\top & \vdots & \operatorname{diag}(\vec{y}) \\
0 & 0 & & 1 & \\
1 & -1 & 0^n & 0 & 0^m
\end{bmatrix} \begin{bmatrix}
\vec{t} \\
b \\
\vec{\xi}
\end{bmatrix} = \begin{bmatrix}
\vec{y} \\
1
\end{bmatrix}\\
& \vec{t} \in \lorentz^{n+2},\; b \in \lorentz^1,\; \xi_i \in \lorentz^1\quad \forall i \in [m]
\end{array} \label{prob:SVM SOCP primal}
\end{equation}
Here, we use the notation $X \in \R^{n\times m}$ for the matrix whose columns are the training examples $\vec{x}^{(i)}$, and $\vec{y} \in \R^m$ for the vector of labels.
This problem has $O(n+m)$ variables, and $O(m)$ conic constraints (i.e. its rank is $r = O(m)$). Therefore, in the interesting case of $m = \Theta(n)$, it can be solved in $\widetilde{O}(\sqrt{n})$ iterations. More precisely, if we consider both the primal and the dual, in total they have $3m+2n+7$ scalar variables and $2m+4$ conic constraints.

In practice (as evidenced by the LIBSVM and LIBLINEAR libraries \cite{chang2011libsvm, fan2008liblinear}), a small modification is made to the formulations \eqref{prob:SVM} and \eqref{prob:LS-SVM}: instead of treating the bias separately, all data points are extended with a constant unit coordinate. In this case, the SOCP formulation remains almost identical, with the only difference being that the constraints $\vec{t} \in \lorentz^{n+2}$ and $b \in \lorentz^1$ are replaced by a single conic constraint $(\vec{t};b) \in \lorentz^{n+3}$. This change allows us to come up with a simple feasible initial solution in our numerical experiments, without going through the homogeneous self-dual formalism of \cite{ye1994hsd}.

Note also that we can solve the LS-SVM problem \eqref{prob:LS-SVM}, by reducing it to a SOCP in a similar manner. In fact, this would have resulted in just $O(1)$ conic constraints, so an IPM would converge to a solution in $\widetilde{O}(1)$ iterations, which is comparable with the result from \cite{rebentrost2014quantum}.

\subsection{Experimental results}
We next present some experimental results to assess the running time parameters and the performance of our algorithm for random instances of SVM. If an algorithm demonstrates a speedup on unstructured instances like these, it is reasonable to extrapolate that the speedup is generic, as it could not have used any special properties of the instance to derive an advantage. For a given dimension $n$ and number of training points $m$, we denote our distribution of random SVMs with $\mathcal{SVM}(n, m, p)$, where $p$ denotes the probability that a datapoint is misclassified by the optimal separating hyperplane. Additionally, for every training set sampled from $\mathcal{SVM}(n, m, p)$, a corresponding test set of size $\lfloor m/3 \rfloor$ was also sampled from the same distribution. These test sets are used to evaluate the generalization error of SVMs trained in various ways.

Our experiments consist of generating roughly 16000 instances of $\mathcal{SVM}(n, 2n, p)$, where $n$ is chosen to be uniform between $2^2$ and $2^9$ and $p$ is chosen uniformly from the discrete set $\{ 0, 0.1, \dots, 0.9, 1 \}$. The instances are then solved using a simulation of Algorithm~\ref{alg:qipm} (with the target duality gap of $\epsilon=0.1$) as well as using the ECOS SOCP solver \cite{domahidi2013ecos} (with the default target duality gap). We simulate the execution of Algorithm \ref{alg:qipm} by implementing the classical IPM and adding noise to the solution of the Newton system \eqref{eq:Newton system}. The noise added to each coordinate is uniform, from an interval selected so that the noisy increment $(\dx, \dy, \ds)$ simulates the outputs of the tomography algorithm with precision determined by Theorem \ref{thm:main}. The SVM parameter $C$ is set to be equal to $1$ in all experiments. Additionally, a separate, smaller experiment with roughly 1000 instances following the same distribution is performed for comparing Algorithm~\ref{alg:qipm} with LIBSVM \cite{chang2011libsvm} using a linear kernel.

The experiments are performed on a Dell Precision 7820T workstation with two Intel Xeon Silver 4110 CPUs and 64GB of RAM, and experiment logs are available \change{at~\cite{FigshareData}}.
\begin{figure}
	\centering
	\begin{minipage}{.49\linewidth}
		\centering
		\includegraphics[width=\linewidth]{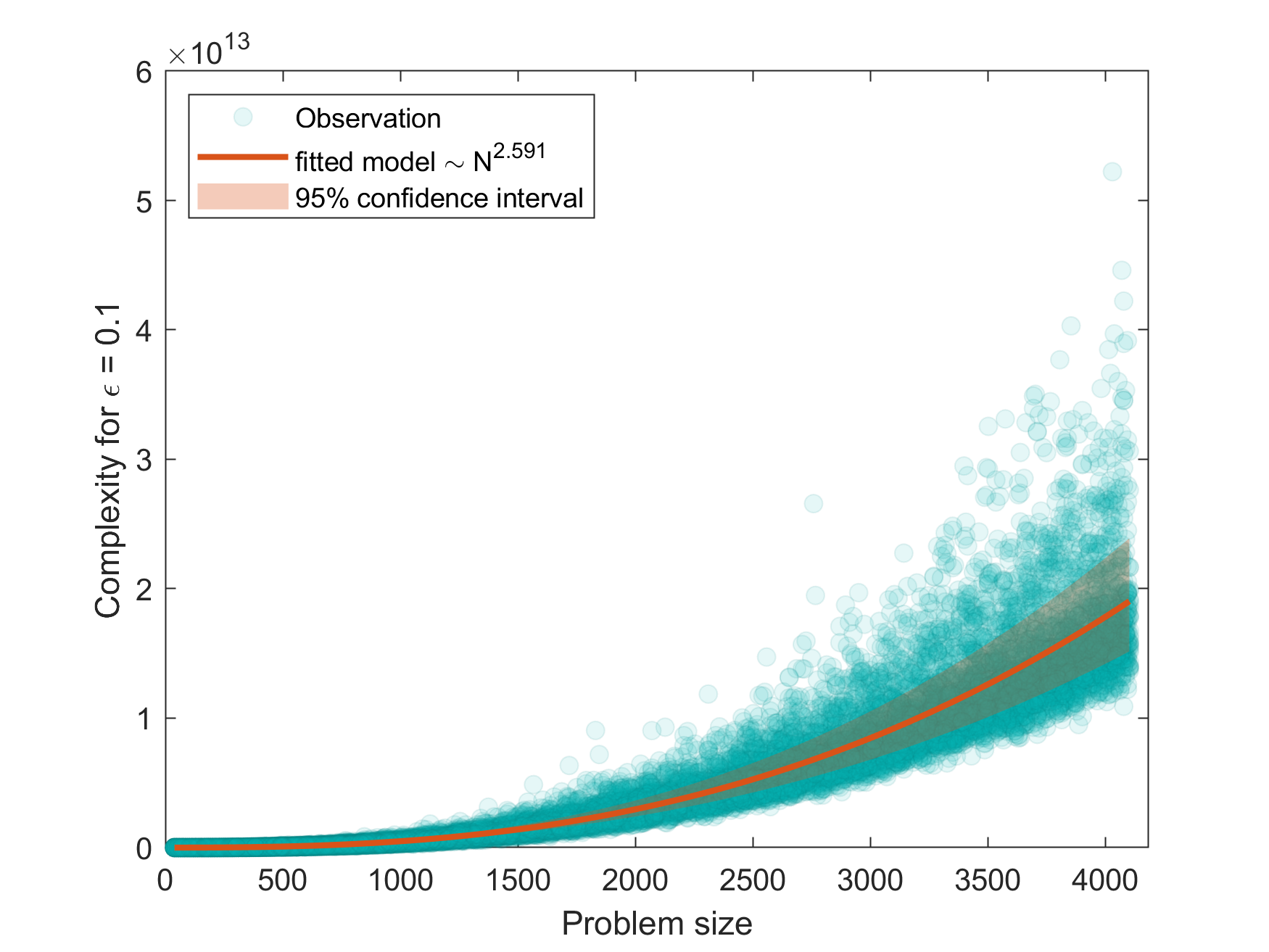}
		\captionof{figure}{Observed complexity of Algorithm~\ref{alg:qipm}, its power law fit, and 95\% confidence interval.}
		\label{fig:complexity}
	\end{minipage}~~
	\begin{minipage}{.49\linewidth}
		\centering
		\includegraphics[width=\linewidth]{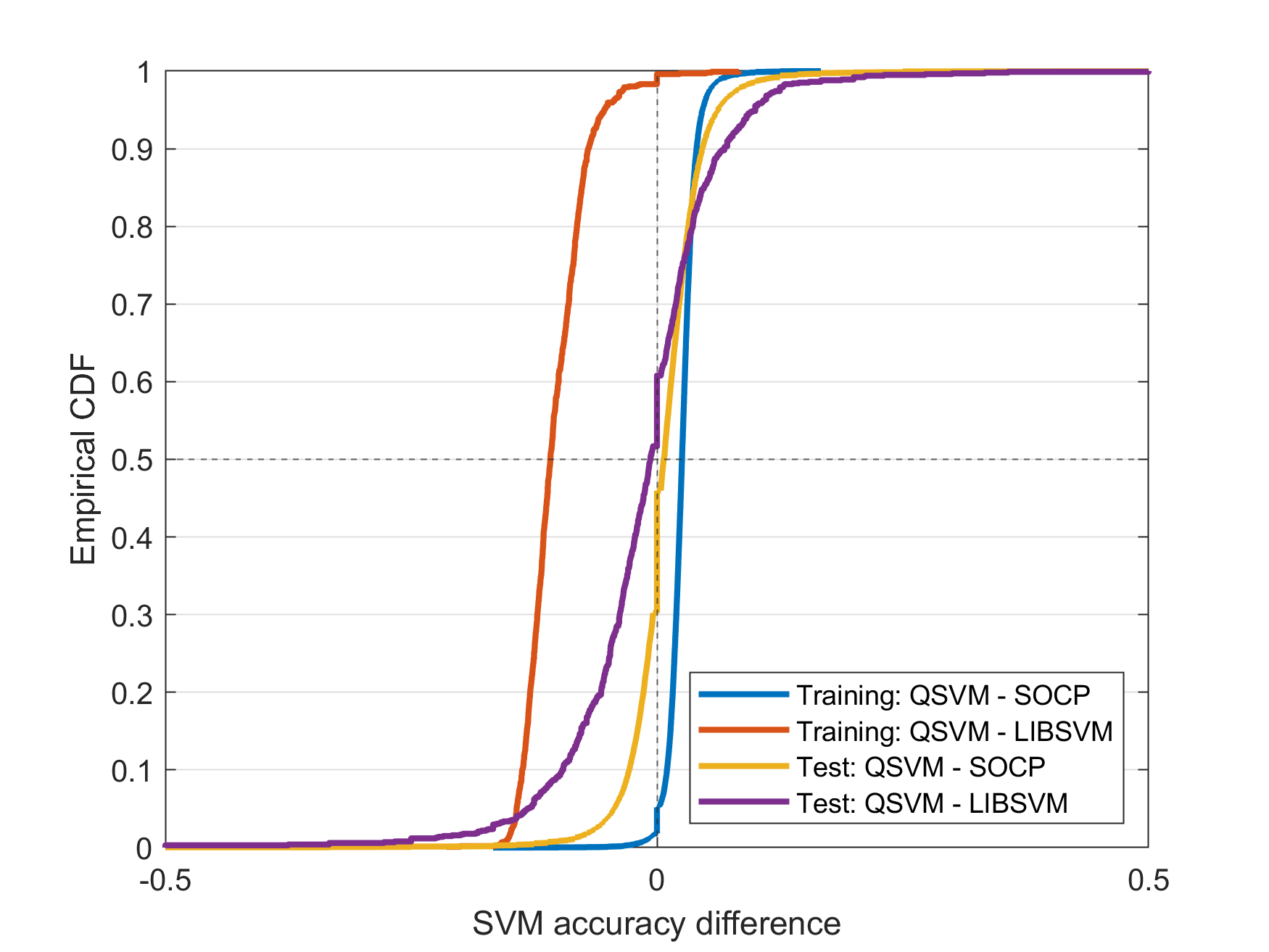}
		\captionof{figure}{Empirical CDF of the difference in accuracy between SVMs trained in different ways.}
		\label{fig:SVM accuracies}
	\end{minipage}
\end{figure}
By finding the least-squares fit of the power law $y = ax^b$ through the observed values of the quantity $\frac{n^{1.5} \kappa \zeta}{\delta^2}$, we obtain the exponent $b = \exponent$, and its 95\% confidence interval $[\exponentlb, \exponentub]$ (this interval is computed in the standard way using Student's $t$-distribution, as described in \cite{neter1996applied}). These observations, the power law, and its confidence interval are show on Figure~\ref{fig:complexity}. Thus, we can say that for random $\mathcal{SVM}(n, 2n, p)$-instances, and fixed $\epsilon=0.1$, the complexity of Algorithm \ref{alg:qipm} scales as $O(n^{\exponent})$. This represents a polynomial improvement over general dense SOCP solvers with complexity $O(n^{\omega+0.5})$. In practice, the polynomial speedup is conserved when compared to ECOS \cite{domahidi2013ecos}, that has a measured running time scaling of $O(n^{3.314})$, with a 95\% confidence interval for the exponent of [3.297, 3.330] (this is consistent with the internal usage of a \cite{strassen1969gaussian}-like matrix multiplication algorithm, with complexity $O(n^{2.807})$). Neglecting constant factors, this gives us a speedup of $10^4$ for $n=10^6$. The results from the LIBSVM solver indicate that the training time with a linear kernel has a complexity of $O(n^{3.112})$, with a 95\% confidence interval for the exponent of [2.799, 3.425]. These results suggest Algorithm~\ref{alg:qipm} retains its advantage even when compared to state-of-the-art specialized classical algorithms.

Additionally, we use the gathered data to verify that the accuracy of our quantum (or approximate) SVM is close to the optimum: Figure~\ref{fig:SVM accuracies} shows that both at train- and at test-time the accuracies of all three classifiers are most of the time within a few percent of each other, with Algorithm~\ref{alg:qipm} often outperforming the exact SOCP SVM classifier.

In conclusion, the performed numerical experiments indicate that Algorithm~\ref{alg:qipm} provides a polynomial speedup for solving SOCPs with low- and medium precision requirements. In particular, for SVM, we achieve a polynomial speedup with no detriment to the quality of the trained classifier.
	\section*{Acknowledgmenets}
This work was partly supported by IdEx Université de Paris ANR-18-IDEX-0001, as well the French National Research Agency (ANR) projects QuBIC and QuDATA.
	
	\bibliographystyle{plainnat}
	\bibliography{bibliography}

\begin{thebibliography}{42}
\providecommand{\natexlab}[1]{#1}
\providecommand{\url}[1]{\texttt{#1}}
\expandafter\ifx\csname urlstyle\endcsname\relax
  \providecommand{\doi}[1]{doi: #1}\else
  \providecommand{\doi}{doi: \begingroup \urlstyle{rm}\Url}\fi

\bibitem[Alizadeh and Goldfarb(2003)]{alizadeh2003second}
F.~Alizadeh and D.~Goldfarb.
\newblock Second-order cone programming.
\newblock \emph{Math. Program.}, 95\penalty0 (1, Ser. B):\penalty0 3--51, 2003.
\newblock ISSN 0025-5610.
\newblock \doi{10.1007/s10107-002-0339-5}.
\newblock ISMP 2000, Part 3 (Atlanta, GA).

\bibitem[Allcock and Hsieh(2020)]{allcock2020quantum}
Jonathan Allcock and Chang-Yu Hsieh.
\newblock A quantum extension of {SVM}-perf for training nonlinear {SVM}s in
  almost linear time.
\newblock \emph{{Quantum}}, 4:\penalty0 342, October 2020.
\newblock ISSN 2521-327X.
\newblock \doi{10.22331/q-2020-10-15-342}.

\bibitem[Arodz and Saeedi(2019)]{arodz2019quantum}
Tomasz Arodz and Seyran Saeedi.
\newblock Quantum sparse support vector machines, 2019.
\newblock URL \url{https://arxiv.org/abs/1902.01879}.

\bibitem[Arora et~al.(2012)Arora, Hazan, and Kale]{arora2012multiplicative}
Sanjeev Arora, Elad Hazan, and Satyen Kale.
\newblock The multiplicative weights update method: a meta-algorithm and
  applications.
\newblock \emph{Theory Comput.}, 8:\penalty0 121--164, 2012.
\newblock \doi{10.4086/toc.2012.v008a006}.

\bibitem[Ben-Tal and Nemirovski(2001)]{ben2001lectures}
Aharon Ben-Tal and Arkadi Nemirovski.
\newblock \emph{Lectures on modern convex optimization}.
\newblock MPS/SIAM Series on Optimization. Society for Industrial and Applied
  Mathematics (SIAM), Philadelphia, PA; Mathematical Programming Society (MPS),
  Philadelphia, PA, 2001.
\newblock ISBN 0-89871-491-5.
\newblock \doi{10.1137/1.9780898718829}.
\newblock Analysis, algorithms, and engineering applications.

\bibitem[Boyd and Vandenberghe(2004)]{boyd2004convex}
Stephen Boyd and Lieven Vandenberghe.
\newblock \emph{Convex optimization}.
\newblock Cambridge University Press, Cambridge, 2004.
\newblock ISBN 0-521-83378-7.
\newblock \doi{10.1017/CBO9780511804441}.

\bibitem[Brand\~{a}o et~al.(2019)Brand\~{a}o, Kalev, Li, Lin, Svore, and
  Wu]{brandao2019quantum}
Fernando G. S.~L. Brand\~{a}o, Amir Kalev, Tongyang Li, Cedric Yen-Yu Lin,
  Krysta~M. Svore, and Xiaodi Wu.
\newblock Quantum {SDP} solvers: large speed-ups, optimality, and applications
  to quantum learning.
\newblock In \emph{46th {I}nternational {C}olloquium on {A}utomata,
  {L}anguages, and {P}rogramming}, volume 132 of \emph{LIPIcs. Leibniz Int.
  Proc. Inform.}, pages Art. No. 27, 14. Schloss Dagstuhl. Leibniz-Zent.
  Inform., Wadern, 2019.
\newblock \doi{10.4230/LIPICS.ICALP.2019.27}.

\bibitem[Brandão et~al.(2020)Brandão, Kueng, and França]{brandao2019faster}
Fernando G. S~L. Brandão, Richard Kueng, and Daniel~Stilck França.
\newblock Faster quantum and classical sdp approximations for quadratic binary
  optimization, 2020.
\newblock URL \url{https://arxiv.org/abs/1909.04613}.

\bibitem[Brandão and Svore(2017)]{brandao2017quantum}
Fernando~G.S.L. Brandão and Krysta~M. Svore.
\newblock Quantum speed-ups for solving semidefinite programs.
\newblock In \emph{58th {A}nnual {IEEE} {S}ymposium on {F}oundations of
  {C}omputer {S}cience---{FOCS} 2017}, pages 415--426. IEEE Computer Soc., Los
  Alamitos, CA, 2017.
\newblock \doi{10.1109/focs.2017.45}.

\bibitem[Bubeck(2015)]{bubeck2015convex}
S{\'{e}}bastien Bubeck.
\newblock Convex optimization: Algorithms and complexity.
\newblock \emph{Foundations and Trends{\textregistered} in Machine Learning},
  8\penalty0 (3-4):\penalty0 231--357, 2015.
\newblock \doi{10.1561/2200000050}.

\bibitem[Chakraborty et~al.(2019)Chakraborty, Gily\'{e}n, and
  Jeffery]{chakraborty2019power}
Shantanav Chakraborty, Andr\'{a}s Gily\'{e}n, and Stacey Jeffery.
\newblock The power of block-encoded matrix powers: improved regression
  techniques via faster {H}amiltonian simulation.
\newblock In \emph{46th {I}nternational {C}olloquium on {A}utomata,
  {L}anguages, and {P}rogramming}, volume 132 of \emph{LIPIcs. Leibniz Int.
  Proc. Inform.}, pages Art. No. 33, 14. Schloss Dagstuhl. Leibniz-Zent.
  Inform., Wadern, 2019.
\newblock \doi{10.4230/LIPICS.ICALP.2019.33}.

\bibitem[Chang and Lin(2011)]{chang2011libsvm}
Chih-Chung Chang and Chih-Jen Lin.
\newblock {LIBSVM}: A library for support vector machines.
\newblock \emph{{ACM} Transactions on Intelligent Systems and Technology},
  2\penalty0 (3):\penalty0 1--27, April 2011.
\newblock \doi{10.1145/1961189.1961199}.

\bibitem[Cohen et~al.(2019)Cohen, Lee, and Song]{cohen2019solving}
Michael~B. Cohen, Yin~Tat Lee, and Zhao Song.
\newblock Solving linear programs in the current matrix multiplication time.
\newblock In \emph{S{TOC}'19---{P}roceedings of the 51st {A}nnual {ACM}
  {SIGACT} {S}ymposium on {T}heory of {C}omputing}, pages 938--942. ACM, New
  York, 2019.
\newblock \doi{10.1145/3313276.3316303}.

\bibitem[Dollar(2005)]{dollar2005iterative}
Hilary Dollar.
\newblock \emph{Iterative linear algebra for constrained optimization}.
\newblock PhD thesis, University of Oxford, 2005.
\newblock URL \url{https://www.numerical.rl.ac.uk/people/hsd/thesismain.pdf}.

\bibitem[Domahidi et~al.(2013)Domahidi, Chu, and Boyd]{domahidi2013ecos}
Alexander Domahidi, Eric Chu, and Stephen Boyd.
\newblock {ECOS}: An {SOCP} solver for embedded systems.
\newblock In \emph{2013 European Control Conference ({ECC})}. {IEEE}, July
  2013.
\newblock \doi{10.23919/ecc.2013.6669541}.

\bibitem[Fan et~al.(2008)Fan, Chang, Hsieh, Wang, and Lin]{fan2008liblinear}
Rong-En Fan, Kai-Wei Chang, Cho-Jui Hsieh, Xiang-Rui Wang, and Chih-Jen Lin.
\newblock {LIBLINEAR}: A library for large linear classification.
\newblock \emph{Journal of machine learning research}, 9\penalty0
  (Aug):\penalty0 1871--1874, 2008.
\newblock URL \url{https://www.jmlr.org/papers/volume9/fan08a/fan08a}.

\bibitem[Gily\'{e}n et~al.(2019{\natexlab{a}})Gily\'{e}n, Arunachalam, and
  Wiebe]{gilyen2019optimizing}
Andr\'{a}s Gily\'{e}n, Srinivasan Arunachalam, and Nathan Wiebe.
\newblock Optimizing quantum optimization algorithms via faster quantum
  gradient computation.
\newblock In \emph{Proceedings of the {T}hirtieth {A}nnual {ACM}-{SIAM}
  {S}ymposium on {D}iscrete {A}lgorithms}, pages 1425--1444. SIAM,
  Philadelphia, PA, 2019{\natexlab{a}}.
\newblock \doi{10.1137/1.9781611975482.87}.

\bibitem[Gily\'{e}n et~al.(2019{\natexlab{b}})Gily\'{e}n, Su, Low, and
  Wiebe]{gilyen2019quantum}
Andr\'{a}s Gily\'{e}n, Yuan Su, Guang~Hao Low, and Nathan Wiebe.
\newblock Quantum singular value transformation and beyond: exponential
  improvements for quantum matrix arithmetics.
\newblock In \emph{S{TOC}'19---{P}roceedings of the 51st {A}nnual {ACM}
  {SIGACT} {S}ymposium on {T}heory of {C}omputing}, pages 193--204. ACM, New
  York, 2019{\natexlab{b}}.
\newblock \doi{10.1145/3313276.3316366}.

\bibitem[Gilyén et~al.(2018)Gilyén, Lloyd, and Tang]{gilyen2018quantum}
András Gilyén, Seth Lloyd, and Ewin Tang.
\newblock Quantum-inspired low-rank stochastic regression with logarithmic
  dependence on the dimension, 2018.
\newblock URL \url{https://arxiv.org/abs/1811.04909}.

\bibitem[Grover(1996)]{grover1996fast}
Lov~K. Grover.
\newblock A fast quantum mechanical algorithm for database search.
\newblock In \emph{Proceedings of the {T}wenty-eighth {A}nnual {ACM}
  {S}ymposium on the {T}heory of {C}omputing ({P}hiladelphia, {PA}, 1996)},
  pages 212--219. ACM, New York, 1996.
\newblock \doi{10.1145/237814.237866}.

\bibitem[Harrow et~al.(2009)Harrow, Hassidim, and Lloyd]{harrow2009quantum}
Aram~W. Harrow, Avinatan Hassidim, and Seth Lloyd.
\newblock Quantum algorithm for linear systems of equations.
\newblock \emph{Physical Review Letters}, 103\penalty0 (15), October 2009.
\newblock \doi{10.1103/physrevlett.103.150502}.

\bibitem[Joachims(2006)]{joachims2006training}
Thorsten Joachims.
\newblock Training linear {SVM}s in linear time.
\newblock In \emph{Proceedings of the 12th ACM SIGKDD International Conference
  on Knowledge Discovery and Data Mining}, KDD '06, page 217–226, New York,
  NY, USA, 2006. Association for Computing Machinery.
\newblock ISBN 1595933395.
\newblock \doi{10.1145/1150402.1150429}.

\bibitem[Karmarkar(1984)]{karmarkar1984new}
N.~Karmarkar.
\newblock A new polynomial-time algorithm for linear programming.
\newblock In \emph{Proceedings of the sixteenth annual {ACM} symposium on
  Theory of computing - {STOC} {\textquotesingle}84}, pages 302--311. {ACM}
  Press, 1984.
\newblock \doi{10.1145/800057.808695}.

\bibitem[Kerenidis and Prakash(2017)]{kerenidis2016quantum}
Iordanis Kerenidis and Anupam Prakash.
\newblock Quantum recommendation systems.
\newblock In Christos~H. Papadimitriou, editor, \emph{8th Innovations in
  Theoretical Computer Science Conference (ITCS 2017)}, volume~67 of
  \emph{Leibniz International Proceedings in Informatics (LIPIcs)}, pages
  49:1--49:21, Dagstuhl, Germany, 2017. Schloss Dagstuhl--Leibniz-Zentrum fuer
  Informatik.
\newblock ISBN 978-3-95977-029-3.
\newblock \doi{10.4230/LIPIcs.ITCS.2017.49}.

\bibitem[Kerenidis and Prakash(2020{\natexlab{a}})]{kerenidis2017quantum}
Iordanis Kerenidis and Anupam Prakash.
\newblock Quantum gradient descent for linear systems and least squares.
\newblock \emph{Phys. Rev. A}, 101:\penalty0 022316, Feb 2020{\natexlab{a}}.
\newblock \doi{10.1103/PhysRevA.101.022316}.

\bibitem[Kerenidis and Prakash(2020{\natexlab{b}})]{kerenidis2018quantum}
Iordanis Kerenidis and Anupam Prakash.
\newblock A quantum interior point method for {LP}s and {SDP}s.
\newblock \emph{ACM Transactions on Quantum Computing}, 1\penalty0 (1), October
  2020{\natexlab{b}}.
\newblock ISSN 2643-6809.
\newblock \doi{10.1145/3406306}.

\bibitem[Lee et~al.(2019)Lee, Song, and Zhang]{lee2019solving}
Yin~Tat Lee, Zhao Song, and Qiuyi Zhang.
\newblock Solving empirical risk minimization in the current matrix
  multiplication time.
\newblock In Alina Beygelzimer and Daniel Hsu, editors, \emph{Proceedings of
  the Thirty-Second Conference on Learning Theory}, volume~99 of
  \emph{Proceedings of Machine Learning Research}, pages 2140--2157, Phoenix,
  USA, 25--28 Jun 2019. PMLR.
\newblock URL \url{http://proceedings.mlr.press/v99/lee19a.html}.

\bibitem[Monteiro and Tsuchiya(2000)]{monteiro2000polynomial}
Renato D.~C. Monteiro and Takashi Tsuchiya.
\newblock Polynomial convergence of primal-dual algorithms for the second-order
  cone program based on the {MZ}-family of directions.
\newblock \emph{Math. Program.}, 88\penalty0 (1, Ser. A):\penalty0 61--83,
  2000.
\newblock ISSN 0025-5610.
\newblock \doi{10.1007/PL00011378}.

\bibitem[Nesterov and Todd(1997)]{nesterov1997self}
Yu.~E. Nesterov and M.~J. Todd.
\newblock Self-scaled barriers and interior-point methods for convex
  programming.
\newblock \emph{Math. Oper. Res.}, 22\penalty0 (1):\penalty0 1--42, 1997.
\newblock ISSN 0364-765X.
\newblock \doi{10.1287/moor.22.1.1}.

\bibitem[Nesterov and Todd(1998)]{nesterov1998primal}
Yu.~E. Nesterov and M.~J. Todd.
\newblock Primal-dual interior-point methods for self-scaled cones.
\newblock \emph{SIAM J. Optim.}, 8\penalty0 (2):\penalty0 324--364, 1998.
\newblock ISSN 1052-6234.
\newblock \doi{10.1137/S1052623495290209}.

\bibitem[Neter et~al.(1996)Neter, Kutner, Nachtsheim, and
  Wasserman]{neter1996applied}
John Neter, Michael~H Kutner, Christopher~J Nachtsheim, and William Wasserman.
\newblock \emph{Applied linear statistical models}, volume~4.
\newblock Irwin Chicago, 1996.

\bibitem[Nielsen and Chuang(2009)]{nielsen2010quantum}
Michael~A. Nielsen and Isaac~L. Chuang.
\newblock \emph{Quantum Computation and Quantum Information}.
\newblock Cambridge University Press, 2009.
\newblock \doi{10.1017/cbo9780511976667}.

\bibitem[Rebentrost et~al.(2014)Rebentrost, Mohseni, and
  Lloyd]{rebentrost2014quantum}
Patrick Rebentrost, Masoud Mohseni, and Seth Lloyd.
\newblock Quantum support vector machine for big data classification.
\newblock \emph{Physical Review Letters}, 113\penalty0 (13), September 2014.
\newblock \doi{10.1103/physrevlett.113.130503}.

\bibitem[Saad(2003)]{saad2003iterative}
Yousef Saad.
\newblock \emph{Iterative methods for sparse linear systems}.
\newblock Society for Industrial and Applied Mathematics, Philadelphia, PA,
  second edition, 2003.
\newblock ISBN 0-89871-534-2.
\newblock \doi{10.1137/1.9780898718003}.

\bibitem[Shor(1994)]{shor1999polynomial}
Peter~W. Shor.
\newblock Algorithms for quantum computation: discrete logarithms and
  factoring.
\newblock In \emph{35th {A}nnual {S}ymposium on {F}oundations of {C}omputer
  {S}cience ({S}anta {F}e, {NM}, 1994)}, pages 124--134. IEEE Comput. Soc.
  Press, Los Alamitos, CA, 1994.
\newblock \doi{10.1109/SFCS.1994.365700}.

\bibitem[Strassen(1969)]{strassen1969gaussian}
Volker Strassen.
\newblock Gaussian elimination is not optimal.
\newblock \emph{Numerische Mathematik}, 13\penalty0 (4):\penalty0 354--356,
  August 1969.
\newblock \doi{10.1007/bf02165411}.

\bibitem[Suykens and Vandewalle(1999)]{suykens1999least}
J.A.K. Suykens and J.~Vandewalle.
\newblock Least squares support vector machine classifiers.
\newblock \emph{Neural Processing Letters}, 9\penalty0 (3):\penalty0 293--300,
  1999.
\newblock \doi{10.1023/a:1018628609742}.

\bibitem[Suykens et~al.(2002)Suykens, Brabanter, Lukas, and
  Vandewalle]{suykens2002weighted}
J.A.K. Suykens, J.~De Brabanter, L.~Lukas, and J.~Vandewalle.
\newblock Weighted least squares support vector machines: robustness and sparse
  approximation.
\newblock \emph{Neurocomputing}, 48\penalty0 (1-4):\penalty0 85--105, October
  2002.
\newblock \doi{10.1016/s0925-2312(01)00644-0}.

\bibitem[Szilagyi et~al.(2021)Szilagyi, Kerenidis, and Prakash]{FigshareData}
Daniel Szilagyi, Iordanis Kerenidis, and Anupam Prakash.
\newblock Quantum {SVM} via {SOCP} experiment logs.
\newblock Mar 2021.
\newblock \doi{10.6084/m9.figshare.11778189.v1}.

\bibitem[van Apeldoorn and Gily\'{e}n(2019)]{van2019improvements}
Joran van Apeldoorn and Andr\'{a}s Gily\'{e}n.
\newblock Improvements in quantum {SDP}-solving with applications.
\newblock In \emph{46th {I}nternational {C}olloquium on {A}utomata,
  {L}anguages, and {P}rogramming}, volume 132 of \emph{LIPIcs. Leibniz Int.
  Proc. Inform.}, pages Art. No. 99, 15. Schloss Dagstuhl. Leibniz-Zent.
  Inform., Wadern, 2019.
\newblock \doi{10.4230/LIPICS.ICALP.2019.99}.

\bibitem[van Apeldoorn et~al.(2017)van Apeldoorn, Gily{\'e}n, Gribling, and
  de~Wolf]{van2017quantum}
Joran van Apeldoorn, Andr{\'a}s Gily{\'e}n, Sander Gribling, and Ronald
  de~Wolf.
\newblock Quantum {SDP}-solvers: Better upper and lower bounds.
\newblock In \emph{2017 {IEEE} 58th Annual Symposium on Foundations of Computer
  Science ({FOCS})}. {IEEE}, October 2017.
\newblock \doi{10.1109/focs.2017.44}.

\bibitem[Ye et~al.(1994)Ye, Todd, and Mizuno]{ye1994hsd}
Yinyu Ye, Michael~J. Todd, and Shinji Mizuno.
\newblock An {$O(\sqrt{n}L)$}-iteration homogeneous and self-dual linear
  programming algorithm.
\newblock \emph{Math. Oper. Res.}, 19\penalty0 (1):\penalty0 53--67, 1994.
\newblock ISSN 0364-765X.
\newblock \doi{10.1287/moor.19.1.53}.

\end{thebibliography}
	
\end{document}